\documentclass[10pt]{article}
\usepackage{mathrsfs}
\usepackage{amssymb} 
\usepackage{color} 
\usepackage{amsthm}
\usepackage{amsmath}
\usepackage{graphicx}
\usepackage[square, comma, sort&compress, numbers]{natbib}

 \newtheorem{thm}{Theorem}[section]
 \newtheorem{cor}[thm]{Corollary}

 \theoremstyle{definition}

 \numberwithin{equation}{section}
 
 \newcommand \del \partial

\makeatletter

\newcommand{\Rmnum}[1]{\expandafter\@slowromancap\romannumeral #1@}
\makeatother
\newtheorem{theorem}{Theorem}[section]
\newtheorem{lemma}[theorem]{Lemma}

\newtheorem{proposition}[theorem]{Proposition}

\theoremstyle{definition}
\newtheorem{definition}[theorem]{Definition}

\newtheorem{remark}[theorem]{Remark}

\begin{document}
\title{The Global Nonlinear Stability
of Self-Gravitating Irrotational Chaplygin Fluids in a FRW Geometry}
\author{Philippe G. LeFLOCH\footnotemark[1] \hskip.15cm and
 Changhua WEI\footnotemark[2]}
\renewcommand{\thefootnote}{\fnsymbol{footnote}}

\footnotetext[1]{Laboratoire Jacques-Louis Lions \& Centre National de la Recherche Scientifique, Universit\'e Pierre et Marie Curie (Paris 6), 4 Place Jussieu, 75252 Paris, France. Email: contact@philippelefloch.org}

\footnotetext[2]{Corresponding author: School of Mathematical Sciences, Fudan University,
Shanghai,  200433,  China. Email: changhuawei1986@gmail.com.
}

\date{}
\maketitle
\allowdisplaybreaks[4]

\begin{abstract}
We analyze the global nonlinear stability of FRW (Friedmann-Robertson-Walker) spacetimes in presence of an irrotational perfect fluid. We assume that the fluid is governed by the so-called (generalized) Chaplygin equation of state $p=-\frac{A^{2}}{\rho^{\alpha}}$ relating the pressure to the mass-energy density, in which $A>0$ and $\alpha\in (0,1]$ are constant. We express the Einstein equations in wave gauge as a systems of coupled nonlinear wave equations and by performing a suitable conformal transformation, we are able to analyze the global behavior of solutions in future timelike directions. We establish that the $(3+1)$-spacetime metric and the mass density and velocity vector describing the evolution of the fluid remain globally close to a reference FRW solution, under small initial data perturbations. Our analysis provides also the precise asymptotic behavior of the perturbed solutions in the future directions.
\end{abstract}

\setcounter{tocdepth}{6}
\tableofcontents
 
\section{Introduction}\label{section:1}

\subsection{Main objective}\label{section:1.1}

Recent cosmological observations predict that our Universe is currently enjoying a phase of  accelerated expansion, which could be described for instance by introducing the notion of dark energy. An analysis based on the Big Bang model reveals that our Universe is spatially flat and consists of 70 percents of dark energy (with negative pressure), while the remaining 30 percents consist of dust matter (i.e. cold dark matter plus baryons), as well as negligible radiation. It has been predicted that the dark energy may be responsible for the present acceleration of our Universe. These predictions from physics rely on several possible theories of dark energy. One possibility is to include a cosmological constant in the Einstein equations, while another quite interesting approach models the matter content as a Chaplygin gas or, more generally a generalized Chaplygin gas (GCG) or a modified Chaplygin gas (MCG). In the past decade, these models have been studied extensively by elementary methods of analysis and via numerical simulations. For a more detailed description of these models, we refer to \cite{B,B-B,D-B,G-K,HF,K-M,K-P,P-J,Y} and the references therein.

In this paper, we consider the nonlinear future stability of self-gravitating fluids governed by the (generalized) Chaplygin equation of state and, therefore, analyze the global existence problem for the Einstein-Euler system.

\begin{definition} A {\bf generalized Chaplygin gas,} by definition, is a perfect fluid governed by the equation of state
\begin{equation}\label{1.1}
p =-\frac{A^{2}}{\rho^{\alpha}},
\end{equation}
relating the pressure $p=p(\rho)$ to the mass-energy density $\rho \geq 0$ of this fluid,
where $A$ is a positive constant and $\alpha \in (0,1]$. This is known as a {\bf Chaplygin gas} when $\alpha=1$.
\end{definition}

The Einstein-Euler system for a generalized Chaplygin gas read as follows:
\begin{equation}\label{1.2}
\aligned
\widetilde{G}^{\mu\nu}& =\widetilde{T}^{\mu\nu},\\
\widetilde{\nabla}_{\mu}\widetilde{T}^{\mu\nu}& =0,
\endaligned
\end{equation}
where
$\widetilde{G}^{\mu\nu}=\widetilde{R}ic^{\,\mu\nu}-\frac{1}{2}\widetilde{R}\widetilde{g}^{\mu\nu}$ is Einstein's curvature tensor of an unknown metric
%\begin{equation}\label{1.3}
$\widetilde{g}=\widetilde{g}_{\mu\nu}dx^{\mu}dx^{\nu}$, 
%\end{equation}
while 
$\widetilde{R}ic_{\mu\nu}$ and $\widetilde{R}$ are the Ricci and scalar curvature of $\widetilde{g}$, respectively, and $\widetilde{\nabla}_{\mu}$ denotes the covariant derivative of $\widetilde{g}$.
Here,
$\widetilde{T}^{\mu\nu}$ denotes the stress energy tensor
\begin{equation}\label{1.4}
\widetilde{T}^{\mu\nu}=(\rho+p)\widetilde{u}^{\mu}\widetilde{u}^{\nu}+p\widetilde{g}^{\mu\nu}\quad 
\end{equation}
where $\rho$ denotes the energy density, $p=p(\rho)$ denotes the pressure and is given by \eqref{1.1}, $\widetilde{u}=(\widetilde{u}^{0},\cdots,\widetilde{u}^{3})$ denotes the unit, future-directed, timelike 4-velocity, $\widetilde{g}^{\mu\nu}$ is the inverse of $\widetilde{g}_{\mu\nu}$.
As is standard, we use Einstein's summation convention, i.e. we sum over repeated lower and upper indices.

We are interested in irrotational full fluids, namely, under the equation of state \eqref{1.1}, fluid such that there exists a potential function $\Psi$ allowing us to express the stress energy tensor in form
\begin{equation}\label{1.5}
\aligned
\widetilde{T}^{\mu\nu}
& =(\rho+p)\widetilde{u}^{\mu}\widetilde{u}^{\nu}+p\widetilde{g}^{\mu\nu}
\\
& =A^{\frac{2}{\alpha+1}}I^{-\frac{\alpha}{\alpha+1}}\left[\frac{\widetilde{h}^{\frac{1-\alpha}{2\alpha}}
\widetilde{\del}^{\mu}\Psi\widetilde{\del}^{\nu}\Psi}{(I-\widetilde{h}^{\frac{\alpha+1}{2\alpha}})^{\frac{1}{\alpha+1}}}
-(I-\widetilde{h}^{\frac{\alpha+1}{2\alpha}})^{\frac{\alpha}{\alpha+1}}\widetilde{g}^{\mu\nu}\right],
\endaligned
\end{equation}
where $I$ is a constant depending only on the initial data of the state variables and $\widetilde{h}=-\widetilde{g}_{\mu\nu}\widetilde{\del}^{\mu}\Psi\widetilde{\del}^{\nu}\Psi\geq 0$. The  derivation of \eqref{1.5} from \eqref{1.1} will be presented in Section \ref{section:2.2}, below.

It is well-known that there exists a particular family of solutions to \eqref{1.2}, that is, the
Friedmann-Robertson-Walker cosmological spacetimes, and
our main result in the present work can be simply described as follows.

\begin{theorem}[The nonlinear future stability of Friedmann-Robertson-Walker spacetimes. Preliminary version]
\label{theorem:1.2}
The FRW solutions to the Einstein-Euler system for a generalized Chaplygin gas are
nonlinearly stable toward the future, when the initial data set prescribed on an initial hypersurface is a small perturbation of this FRW solution.
\end{theorem}

A more technical version of the above theorem will be provided below after we will introduce some necessary notations. At this junction, we want to recall some eariler work on the Einstein-Euler system, especially for the expanding spacetimes of interest in the present work.

\subsection{Background on the problem}\label{section:1.2}

One natural problem one may ask is whether or not our universe dominated by above models is stable. Due to the importance of this problem in mathematics and physics, it attracts a lot of attention in general relativity and great improvements have been made on the Einstein-Euler equations with positive cosmological constant in recent years. We will describe these improvements clearly below. While we still know nothing on the future nonlinear stability of our universe dominated by Chaplygin gas (or GCG, MCG) in $1+3$ spacetime dimensions, that is one motivation for us to study this problem.

The Einstein-Euler equations with positive cosmological constant read
\begin{equation}\label{1.6}
\aligned
\widetilde{G}^{\mu\nu}&=\widetilde{T}^{\mu\nu}-\Lambda\widetilde{g}^{\mu\nu},\\
\widetilde{\nabla}_{\mu}\widetilde{T}^{\mu\nu}&=0,
\endaligned
\end{equation}
where $\Lambda>0$ denotes the positive cosmological constant.
One advance for equation \eqref{1.6} with metric \eqref{1.3} and stress energy tensor \eqref{1.4} is focused on the fluid with equation of state
\begin{equation}\label{1.7}
p=C_{s}^{2}\rho,
\end{equation}
where $C_{s}\geq 0$ denotes the speed of sound. When $C_{s}=0$, it is called the dust universe, and when $C_{s}=\sqrt{\frac{1}{3}}$, it is used to describe the radiation universe.

There are mainly two approaches to deal with the Einstein-Euler system \eqref{1.6} mathematically: working with the given spacetime metric, or alternatively working with a conformally equivalent metric. Both approaches play an important role in the mathematical theory of general relativity. The well-known family of Friedmann-Lema\^{\i}tre-Robertson-Walker (FLRW) solutions to \eqref{1.6} represents a homogeneous, fluid filled universe that is undergoing accelerated expansion. Rodnianski and Speck \cite{R-S} established the future global stability of a class of FLRW solutions under the assumption of zero vorticity and  $0<C_{s}<\sqrt{\frac{1}{3}}$. The so-called wave gauge approach  in \cite{R-S} was first used by Ringstr\"{o}m in \cite{R} who treated scalar fields without positive cosmological constant:  $\widetilde{T}^{\mu\nu}=\widetilde{\del}^{\mu}\Psi\widetilde{\del}^{\nu}\Psi-[\frac{1}{2}\widetilde{g}_{\mu\nu}\widetilde{\del}^{\mu}\Psi\widetilde{\del}^{\nu}\Psi+V(\Psi)]\widetilde{g}^{\mu\nu}$. The presence of $V(\Psi)$ plays an analogue role as $\Lambda$, under the assumption that $V(0)>0,\, V^{'}(0)=0,\,V^{''}(0)>0$. The main observation in these papers is two-fold: on one hand, the Einstein-nonlinear scalar field system can be formulated as a system of nonlinear wave equations provided one introduces generalized wave coordinates, inspired by the standard wave coordinates used earlier (for instance in Lindblad and Rodnianski \cite{L-R} for the vacuum case, revisited recently in LeFloch and Ma \cite{LM});
on the other hand, the problem under consideration describes the expansion of the universe, and the expansion provides one dispersive terms, which leads to exponential decay for solutions. Later, Speck (and collaborator) \cite{H-S,S} proved that above nonlinear stability result remains true even for a fluid with non-vanishing vorticity when $0\leq C_{s}< \sqrt{\frac{1}{3}}$. When $C_{s}=\sqrt{\frac{1}{3}}$, L\"{u}bbe and Valiente Kroon \cite{L-V} have shown the
desired stability property by relying on Friedrich's conformal method \cite{F1,F2} ---an approach entirely different from the method in \cite{H-S,R-S,S}. More recently, a very efficient method was proposed by Oliynyk \cite{O}, which combine the conformal method with wave coordinates in order to handle the case $0<C_{s}\leq \sqrt{\frac{1}{3}}$ with non-vanishing vorticity. One  advantage of the latter method is that, under a conformal transformation, the whole Einstein-Euler system can be turned into a symmetric hyperbolic system (with singular terms) and solutions defined on finite interval of time. The singular terms enjoy good positivity properties and, by a standard energy estimate, one can then get the global nonlinear stability of a family of  FLRW solutions and established the asymptotic behavior of perturbed solutions in the far future.  Finally, we recall that, in the regime $C_{s}>\sqrt{\frac{1}{3}}$, Rendall \cite{Ren} has found some evidence for instability.

\subsection{Nonlinear stability property in wave coordinates}\label{section:1.3}

We thus study the Einstein-Euler system \eqref{1.2} with the stress energy tensor \eqref{1.5}.
In order to describe our ideas and main results clearly, we must fix our notations at first. The spacetime that we consider are of the form $(0,1]\times\mathbb T^{3}$. For the coordinates, we use $x^{i}$ (i=1,2,3) to denote the spacial coordinates and use $x^{0}=\tau$ to denote the time coordinate. We always use the Greek indices to denote the spacetime coordinates that run form 0 to 3 and Latin indices to denote the spacial coordinates that run from 1 to 3. In this system of coordinates, when we say the fluid velocity is future directed, we mean that
\begin{equation}\label{1.8}
\widetilde{u}^{0}<0.
\end{equation}
As Oliynyk in \cite{O}, we do not consider the original metric $\widetilde{g}$ directly, but instead consider the conformally transformed metric
\begin{equation}\label{1.9}
g_{\mu\nu}=e^{-2\Phi}\widetilde{g}_{\mu\nu},\quad \text{or}\quad g^{\mu\nu}=e^{2\Phi}\widetilde{g}^{\mu\nu},
\end{equation}
where
\begin{equation}\label{1.10}
\Phi=-\ln(\tau).
\end{equation}
Under the conformal transformation \eqref{1.9} and \eqref{1.10}, the equations \eqref{1.2} that we consider in this paper is the following Cauchy problem
\begin{equation}\label{1.11}
\aligned
G^{\mu\nu} &=T^{\mu\nu}:=e^{4\Phi}\widetilde{T}^{\mu\nu}+2(\nabla^{\mu}\nabla^{\nu}\Phi-\nabla^{\mu}\Phi\nabla^{\nu}\Phi)-(2\Box_{g}\Phi+|\nabla\Phi|_{g}^{2})g^{\mu\nu},\\
\nabla_{\mu}\widetilde{T}^{\mu\nu} &=-6\widetilde{T}^{\mu\nu}\nabla_{\mu}\Phi+g_{\kappa\lambda}\widetilde{T}^{\kappa\lambda}g^{\mu\nu}\nabla_{\mu}\Phi,\\
%\tau=1:\quad
g^{\mu\nu}|_{\tau=1} &=g_{0}^{\mu\nu}(x),\quad \del_{\tau}g^{\mu\nu}|_{\tau=1}=g_{1}^{\mu\nu}(x),\quad \del_{\mu}\Psi|_{\tau=1}=m_{\mu}(x).
\endaligned
\end{equation}

\begin{remark}\label{remark:1.3}
Above initial data set $(g_{0}^{\mu\nu}(x),g_{1}^{\mu\nu}(x),m_{\mu}(x))$ can not be chosen arbitrarily. They must satisfy the Gauss-Codazzi equations, which are equivalent to $(G^{\mu0}-T^{\mu0})|_{\tau=1}=0$. Furthermore, they will also satisfy the wave coordinates condition $Z^{\mu}|_{\tau=1}=0$, the precise definition of $Z^{\mu}$ can be found in Section \ref{section:3.1}.
\end{remark}

We know that there exist a family of FRW solutions $(\widetilde{\eta},\overline{\Psi}(\tau))$ to the original Einstein-Euler equations of GCG \eqref{1.2}. $\widetilde{\eta}$ takes the following form
\begin{equation*}
\widetilde{\eta}=\frac{1}{\tau^{2}}\left(-\frac{1}{w^{2}(\tau)}d\tau^{2}+\sum_{i=1}^{3}(dx^{i})^{2}\right),
\end{equation*}
where
\begin{equation*}
\tau=\frac{1}{a(t)}, \quad w=\frac{1}{a(t)}\frac{da(t)}{dt}:=\frac{\dot{a}(t)}{a(t)},
\end{equation*}
for some scale factor $a(t)$ with $t\in[0,+\infty)$. Under conformal transformation, the conformal metric $g$ can be seen as small perturbations to the conformal background metric $\eta$ which is given by
\begin{equation*}
\eta=-\frac{1}{w^{2}}d\tau^{2}+\sum_{i=1}^{3}(dx^{i})^{2}.
\end{equation*}
Define the densitized three-metric
$
\textbf{g}^{ij}=det(\check{g}_{lm})^{\frac{1}{3}}g^{ij},
$
where
$
\check{g}_{lm}=(g^{lm})^{-1},
$
and introduce the variable
\begin{equation*}
\textbf{q}=g^{00}-\eta^{00}+\frac{\eta^{00}}{3}\ln(\det(g^{ij})).
\end{equation*}
With above notations, our main result can be described as follows.

\begin{theorem}[The nonlinear future stability of Friedmann-Robertson-Walker spacetimes.  Statement in wave coordinates]\label{theorem:1.4}
Suppose $k\geq3$, $g_{0}^{\mu\nu}(x)\in H^{k+1}(\mathbb T^{3})$, $g_{1}^{\mu\nu}(x),\,m_{\mu}(x)\in H^{k}(\mathbb T^{3})$ for all $x\in\mathbb T^{3}$.
Then there exists a small parameter $\epsilon>0$, such that if the initial data sets satisfy the constraint equations of Remark 1.3 and
$$
\|g_{0}^{\mu\nu}-\eta^{\mu\nu}(1)\|_{H^{k+1}}+\|g_{1}^{\mu\nu}-\del_{\tau}\eta^{\mu\nu}(1)\|_{H^{k}}+\|m_{\mu}(x)-\del_{\mu}\overline{\Psi}(1)\|_{H^{k}}<\epsilon,
$$
then there exists a unique classical solution $g^{\mu\nu},\,\Psi\in C^{2}((0,1]\times\mathbb T^{3})$ to the conformal Einstein-Euler system of GCG \eqref{1.11} and it has the following regularity
$$
g^{\mu\nu}\in C^{0}((0,1],H^{k+1}(\mathbb T^{3}))\cap C^{0}([0,1],H^{k}(\mathbb T^{3}))\cap C^{1}((0,1],H^{k}(\mathbb T^{3}))\cap C^{1}([0,1],H^{k-1}(\mathbb T^{3}))
$$
and
$$
\del_{\mu}\Psi\in C^{0}((0,1],H^{k}(\mathbb T^{3}))\cap C^{0}([0,1], H^{k-1}(\mathbb T^{3})).
$$
The solution also satisfies
$$
\|g^{\mu\nu}(\tau)-\eta^{\mu\nu}(\tau)\|_{H^{k+1}}+\|g_{\kappa}^{\mu\nu}(\tau)-\del_{\kappa}\eta^{\mu\nu}(\tau)\|_{H^{k}}
+\|\del_{\mu}\Psi(\tau)-\del_{\mu}\overline{\Psi}(\tau)\|_{H^{k}}<C\epsilon,
$$
for some positive constants $C$.
Moreover, there exists $\gamma^{\mu}\in H^{k-1}(\mathbb T^{3})$, such that the solution for all $\tau\in[0,1]$ satisfies
\begin{eqnarray*} 
\|\del_{\tau}g^{0\mu}(\tau)-2\tau^{-1}(g^{0\mu}(\tau)-\eta^{0\mu}(\tau))+\gamma^{\mu}\|_{H^{k-1}}&\leq&C\epsilon\tau,\\
\|\del_{\tau}g^{0\mu}(\tau)-\tau^{-1}(g^{0\mu}(\tau)-\eta^{0\mu}(\tau))\|_{H^{k-1}}+\|\del_{i}g^{0\mu}(\tau)\|_{H^{k-1}}&\leq&-C\epsilon\tau\ln{\tau},\\
\|q(\tau)-q(0)\|_{H^{k}}+\|\del_{\tau}q(\tau)\|_{H^{k-1}}&\leq&C\epsilon\tau,\\
\|\textbf{g}^{ij}(\tau)-\textbf{g}^{ij}(0)\|_{H^{k}}+\|\del_{\tau}\textbf{g}^{ij}(\tau)\|_{H^{k-1}}&\leq&C \tau,\\
\|\del_{\mu}\Psi(\tau)-\del_{\mu}\Psi(0)\|_{H^{k-1}}&\leq&C\epsilon\tau,\\
\|\del_{\tau}\overline{\Psi}(0)-\del_{\tau}\Psi(0)\|_{H^{k-1}}&\leq&C\epsilon,\\
\|\del_{i}\Psi(0)\|_{H^{k-1}}&\leq&C\epsilon,
\end{eqnarray*}
with, furthermore,
$$
p_{\rho}(0)= \alpha.
$$
\end{theorem}

\begin{remark}
1. The above results show that in the future of the Einstein-Euler equations of GCG, the speed of sound is $\sqrt{\alpha}$, especially for Chaplygin gas, it means that the speed of sound is equal to the speed of light. This phenomenon is very interesting and quite different from the problem of a fluid moving in a Universe with positive cosmological constant $\Lambda$.

2.  When $\alpha=1$, the Euler system is linearly degenerate, while when $0<\alpha<1$, this system is genuinely nonlinear and shocks generally form in finite time, no matter how small the perturbations are in the flat Minkowski spacetime. Our result shows that the spacetime expansion stabilize the fluid and prevent the formation of shocks.
\end{remark}

The main idea of the proof is turning system \eqref{1.11} into a symmetric hyperbolic system under appropriate wave coordinates $Z^{\mu}=0$. Here we would like to see some differences between our paper and Oliynyk's and others' mentioned above. Firstly, we consider the nonlinear future stability of the nontrivial FRW solutions, thus, we need to choose different coordinates, which can be seen as a generalization of Oliynyk's work; second, for irrotational fluids, we have to choose an appropriate conformal factor in order to solve the difficulty brought by the degeneracy of the enthalpy $\sqrt{\widetilde{h}}$. 

An outline of this paper is as follows. In Section \ref{section:2.1}, we give some preliminaries on the notations and norms used in this paper. Sections \ref{section:2.2} and \ref{section:2.3} are aimed at giving the detailed formulation of the problem and studying the properties of the FRW background solutions. Section \ref{section:3} is the main part of the whole paper, we present our choice of coordinates in Section \ref{section:3.1}. In Sections \ref{section:3.2}, \ref{section:3.3} and \ref{section:3.4}, we turn the whole system into a symmetric hyperbolic system and analyze the structure of this system. Sections \ref{section:3.5} and \ref{section:3.6} contain  the proof of the main results.

%==================================================================

\section{Formulation of the problem}\label{section:2}

\subsection{Notation}\label{section:2.1}

Greek indices $\mu$ range from 0 to 3, while Latin indices $i$ range from 1 to 3. Repeated lower and upper indices means summation with their corresponding metric. For the metrics in this paper, we use $\widetilde{g}$ and $\widetilde{\eta}$ to denote the original metric and original background metric respectively; we also use $g$ and $\eta$ to denote the conformal metric and the conformal background metric. $\widetilde{\Gamma}$, $\overline{\widetilde{\Gamma}}$, $\Gamma$ and $\overline{\Gamma}$ denote the Christoffel symbols with respect to $\widetilde{g}$, $\widetilde{\eta}$, $g$ and $\eta$, respectively, similar conventions are used for the curvature tensors $\widetilde{R}$, $\overline{\widetilde{R}}$, $R$, $\overline{R}$ and the norms
$\widetilde{h}$, $\overline{\widetilde{h}}$, $h$ and $\overline{h}$, whose square root denotes the physical quantity ``enthalpy''.

For convenience, we use $A\sim B$ to denote the equivalence relationship between $A$ and $B$, which means that there exists a positive constant $C>1$, such that $\frac{A}{C}\leq B\leq CA$. $\widetilde{\del}_{\mu}=\widetilde{\del}_{x^{\mu}}$ and $\del_{\mu}=\del_{x^{\mu}}$ are the partial derivatives of the original spacetime and the conformal spacetime. Similar definitions are used for the covariant derivatives $\widetilde{\nabla}$ and $\nabla$.

For a function $u(t,x)$, we define the following standard sobolev norms
$$
\aligned
\|u(t,x)\|_{L^{2}(\mathbb T^{n})}:=& \left(\int_{\mathbb T^{n}}|u(t,x)|^{2}dx\right)^{\frac{1}{2}},
\\
\|u(t,x)\|_{H^{k}(\mathbb T^{n})}:=& \sum_{I=0}^{k}\|D^{I}u(t,x)\|_{L^{2}(\mathbb T^{n})},
\endaligned
$$
and
$
\|u(t,x)\|_{L^{\infty}(\mathbb T^{3})}:=\text{ess} \sup_{x\in\mathbb T^{3}}|u(t,x)|.
$

 \subsection{Stress energy tensor for irrotational fluids}\label{section:2.2}

 In this section, we focus on the derivation of the irrotational energy momentum tensor $\widetilde{T}^{\mu\nu}$.  
For isentropic fluids, the pressure is given by
\begin{equation}\label{3.1}
p=n \frac{\del \rho}{\del n}-\rho,
\end{equation}
where $n$ denotes the number of particles per unit volume. We also have 
\begin{equation*}
\widetilde{\nabla}_{\mu}(n \widetilde{u}^{\mu})=0.
\end{equation*}
In this paper, we consider the generalized Chaplygin gas, whose equation of state is given by
\begin{equation}
p=-\frac{A^{2}}{\rho^{\alpha}},
\end{equation}
where $A$ is a positive constant and $0<\alpha\leq 1$.

By solving a simple ODE \eqref{3.1}, we have
$
\rho^{\alpha+1}-A^{2}=J(1)n^{\alpha+1},
$
where $J(1)=\frac{\rho^{\alpha+1}(1)-A^{2}}{n^{\alpha+1}(1)}$ depends only on the initial data of the state variables $(\rho,n)$ at $\tau=1$.
Define the enthalpy in the original spacetime by
$$
\sqrt{\widetilde{h}}=\frac{\rho+p}{n}=J^{\frac{1}{\alpha+1}}(1)\left(1-\frac{A^{2}}{\rho^{\alpha+1}}\right)^{\frac{\alpha}{\alpha+1}},
$$
then
$$
\rho=\left[\frac{A^{2}J^{\frac{1}{\alpha}}(1)}{J^{\frac{1}{\alpha}}(1)-\widetilde{h}^{\frac{\alpha+1}{2\alpha}}}\right]^{\frac{1}{\alpha+1}}.
$$
Denote $J^{\frac{1}{\alpha}}(1)=I$, then
\begin{equation}\label{3.2}
\rho=\frac{A^{\frac{2}{\alpha+1}}I^{\frac{1}{\alpha+1}}}{(I-\widetilde{h}^{\frac{\alpha+1}{2\alpha}})^{\frac{1}{\alpha+1}}},
\end{equation}
and
\begin{equation}\label{3.3}
p=-\frac{A^{2}}{\rho^{\alpha}}=-A^{\frac{2}{\alpha+1}}I^{-\frac{\alpha}{\alpha+1}}(I-\widetilde{h}^{\frac{\alpha+1}{2\alpha}})^{\frac{\alpha}{\alpha+1}}.
\end{equation}
By above two equalities \eqref{3.2} and \eqref{3.3}, we can easily get
$$
p+\rho=\frac{A^{\frac{2}{\alpha+1}}I^{-\frac{\alpha}{\alpha+1}}\widetilde{h}^{\frac{\alpha+1}{2\alpha}}}
{(I-\widetilde{h}^{\frac{\alpha+1}{2\alpha}})^{\frac{1}{\alpha+1}}},
$$
and
$$
\rho-p=A^{\frac{2}{\alpha+1}}I^{-\frac{\alpha}{\alpha+1}}\frac{2I-\widetilde{h}^{\frac{\alpha+1}{2\alpha}}}{(I-\widetilde{h}^{\frac{\alpha+1}{2\alpha}})^{\frac{1}{\alpha+1}}}.
$$
Assume that the fluid is irrotational, then there locally exists a potential function $\Psi(\tau,x)$ such that
$$
\widetilde{u}^{\mu}=-\frac{\widetilde{\del}^{\mu}\Psi}{\sqrt{\widetilde{h}}}.
$$
Based on the  normalization condition
$
\widetilde{g}_{\mu\nu}\widetilde{u}^{\mu}\widetilde{u}^{\nu}=-1,\, \text{and}\, \widetilde{u}^{0}<0,
$
in our spacetime, we have
$$
\widetilde{h}=-\widetilde{g}_{\mu\nu}\widetilde{\del}^{\mu}\Psi\widetilde{\del}^{\nu}\Psi\geq0,\quad \widetilde{\del}^{\tau}\Psi\geq0.
$$
Thus, we can define our stress energy tensor as follows
\begin{equation}\label{3.4}
\widetilde{T}^{\mu\nu}=(\rho+p)\widetilde{u}^{\mu}\widetilde{u}^{\nu}+p\widetilde{g}^{\mu\nu}=A^{\frac{2}{\alpha+1}}I^{-\frac{\alpha}{\alpha+1}}\left[\frac{\widetilde{h}^{\frac{1-\alpha}{2\alpha}}
\widetilde{\del}^{\mu}\Psi\widetilde{\del}^{\nu}\Psi}{(I-\widetilde{h}^{\frac{\alpha+1}{2\alpha}})^{\frac{1}{\alpha+1}}}
-(I-\widetilde{h}^{\frac{\alpha+1}{2\alpha}})^{\frac{\alpha}{\alpha+1}}\widetilde{g}^{\mu\nu}\right].
\end{equation}
At the same time, we have
\begin{equation}\label{3.5}
\widetilde{T}=\widetilde{g}_{\mu\nu}\widetilde{T}^{\mu\nu}=-\rho-p+4p=A^{\frac{2}{\alpha+1}}I^{-\frac{\alpha}{\alpha+1}}\frac{-4I+3\widetilde{h}^{\frac{\alpha+1}{2\alpha}}}
{(I-\widetilde{h}^{\frac{\alpha+1}{2\alpha}})^{\frac{1}{\alpha+1}}}.
\end{equation}

\begin{remark}\label{remark:2.1}
1. It is easy to choose appropriate initial data for the state variables $(\rho,n)$ such that $I-\widetilde{h}^{\frac{\alpha+1}{2\alpha}}>0$ to ensure the hyperbolicity of the fluid equation.

2. For the existence of the potential function $\Psi$ and further results concerning analysis of irrotational fluids, we refer to
 \cite{LS, R-S,V}.
\end{remark}

%-------------------------------------------------------------------------------------------------------------------

\subsection{The class of Friedmann-Robertson-Walker spacetimes}\label{section:2.3}
In this section, we study some properties of the background FRW solutions. At first,
we are ready to give the detailed information of our spacetime. The metric endowed by the original spacetime $(0,1]\times\mathbb T^{3}$ is
\begin{equation}\label{3.6}
\widetilde{g}=\widetilde{g}_{\mu\nu}dx^{\mu}dx^{\nu},
\end{equation}
which can be seen as the perturbation of the following metric
\begin{equation}\label{3.7}
\widetilde{\eta}=\frac{1}{\tau^{2}}\left(-\frac{1}{w^{2}(\tau)}d\tau^{2}+\sum_{i=1}^{3}(dx^{i})^{2}\right),
\end{equation}
where $\tau=\frac{1}{a(t)}$ and $w(t)=\frac{\dot{a}(t)}{a(t)}$, for some scale factor $a(t)$ with $t\in[0,+\infty)$. Obviously
$$
d\tau=-\frac{1}{a^{2}(t)}\dot{a}(t)=-\tau w(t)dt.
$$
Note in passing that the metric \eqref{3.7} is equivalent to
$$
ds^{2}=-dt^{2}+a^{2}(t)\sum_{i=1}^{3}(dx^{i})^{2},
$$
which is the original model studied in physics.

As discussed in Section \ref{section:1.3}, we do not work with the spacetime metric \eqref{3.6} directly but instead use the conformally transformed metric
\begin{equation}\label{3.8}
g_{\mu\nu}=e^{-2\Phi}\widetilde{g}_{\mu\nu},
\end{equation}
where $\Phi=-\ln(\tau)$. We will consider above $g$ \eqref{3.8} as a small perturbation of the following metric
\begin{equation}\label{3.9}
\eta=-\frac{1}{w^{2}(\tau)}d\tau^{2}+\sum_{i=1}^{3}(dx^{i})^{2},
\end{equation}
which is the conformal transformation of \eqref{3.7}.
Under above conformal transformations, the Einstein equations are equivalent to
\begin{equation}\label{3.10}
\begin{cases}
G^{\mu\nu}=T^{\mu\nu}:=e^{4\Phi}\widetilde{T}^{\mu\nu}+2(\nabla^{\mu}\nabla^{\nu}\Phi-\nabla^{\mu}\Phi\nabla^{\nu}\Phi)-(2\Box_{g}\Phi+|\nabla\Phi|_{g}^{2})g^{\mu\nu},\\
\nabla_{\mu}\widetilde{T}^{\mu\nu}=-6\widetilde{T}^{\mu\nu}\nabla_{\mu}\Phi+g_{\kappa\lambda}\widetilde{T}^{\kappa\lambda}g^{\mu\nu}\nabla_{\mu}\Phi.
\end{cases}
\end{equation}
Expanding the first equation of \eqref{3.10}, we have by inserting $\widetilde{T}^{\mu\nu}$ defined by \eqref{3.4}
\begin{eqnarray}\label{3.11}
-2R^{\mu\nu}&=&-4\nabla^{\mu}\nabla^{\nu}\Phi+4\nabla^{\mu}\Phi\nabla^{\nu}\Phi
-2[\Box_{g}\Phi+2|\nabla\Phi|^{2}_{g}\\
&+&\frac{1}{2}A^{\frac{2}{\alpha+1}}I^{-\frac{\alpha}{\alpha+1}}\frac{2I-\widetilde{h}^{\frac{\alpha+1}{2\alpha}}}
{(I-\widetilde{h}^{\frac{\alpha+1}{2\alpha}})^{\frac{1}{\alpha+1}}}e^{2\Phi}]g^{\mu\nu}
-2e^{4\Phi}\frac{A^{\frac{2}{\alpha+1}}I^{-\frac{\alpha}{\alpha+1}}\widetilde{h}^{\frac{1-\alpha}{2\alpha}}\widetilde{\del}^{\mu}\Psi\widetilde{\del}^{\nu}\Psi}{
(I-\widetilde{h}^{\frac{\alpha+1}{2\alpha}})^{\frac{1}{\alpha+1}}},\nonumber
\end{eqnarray}
or equivalently
\begin{eqnarray}\label{3.12}
-2R_{\mu\nu}&=&-4\nabla_{\mu}\nabla_{\nu}\Phi+4\nabla_{\mu}\Phi\nabla_{\nu}\Phi
-2[\Box_{g}\Phi+2|\nabla\Phi|^{2}_{g}\\
&+&\frac{1}{2}A^{\frac{2}{\alpha+1}}I^{-\frac{\alpha}{\alpha+1}}\frac{2I
-\widetilde{h}^{\frac{\alpha+1}{2\alpha}}}{(I-\widetilde{h}^{\frac{\alpha+1}{2\alpha}})^{\frac{1}{\alpha+1}}}e^{2\Phi}]g_{\mu\nu}
-2\frac{A^{\frac{2}{\alpha+1}}I^{-\frac{\alpha}{\alpha+1}}\widetilde{h}^{\frac{1-\alpha}{2\alpha}}\widetilde{\del}_{\mu}\Psi\widetilde{\del}_{\nu}\Psi}{
(I-\widetilde{h}^{\frac{\alpha+1}{2\alpha}})^{\frac{1}{\alpha+1}}}.\nonumber
\end{eqnarray}

Now we consider the background FRW solution to \eqref{3.10}. Assume that there exists $\Psi(\tau,x)=\overline{\Psi}(\tau)$ satisfies \eqref{3.10} with background metric \eqref{3.9}. Then it is easy to see that
$$
\overline{\widetilde{h}}:=-\widetilde{\eta}_{00}\widetilde{\del}^{\tau}\overline{\Psi}\widetilde{\del}^{\tau}\overline{\Psi}=\frac{1}{\tau^{2}w^{2}}(\widetilde{\del}^{\tau}\overline{\Psi})^{2}.
$$
$$
\widetilde{T}^{00}=A^{\frac{2}{\alpha+1}}I^{-\frac{\alpha}{\alpha+1}}\left[\frac{\overline{\widetilde{h}}^{\frac{\alpha+1}{2\alpha}}\tau^{2}w^{2}}{(I-\overline{\widetilde{h}}
^{\frac{\alpha+1}{2\alpha}})^{\frac{1}{\alpha+1}}}
+(I-\overline{\widetilde{h}}^{\frac{\alpha+1}{2\alpha}})^{\frac{\alpha}{\alpha+1}}\tau^{2}w^{2}\right],
$$
$$
\widetilde{T}^{0i}=0,
$$
and
$$
\widetilde{T}^{ij}=A^{\frac{2}{\alpha+1}}I^{-\frac{\alpha}{\alpha+1}}
\left[-(I-\overline{\widetilde{h}}^{\frac{\alpha+1}{2\alpha}})^{\frac{\alpha}{\alpha+1}}\tau^{2}\right]\delta^{ij}.
$$
Furthermore, we have
$$
\overline{\Gamma}^{\lambda}_{\mu\nu}=\frac{1}{2}\eta^{00}(\del_{\tau}\eta_{00})\delta_{0}^{\lambda}\delta_{\mu}^{0}\delta_{\nu}^{0}
=-\frac{\del_{\tau}w}{w}\delta_{0}^{\lambda}\delta_{\mu}^{0}\delta_{\nu}^{0},
$$
then
$$
\overline{R}_{\mu\nu}=\del_{\alpha}\overline{\Gamma}^{\alpha}_{\mu\nu}-\del_{\mu}\overline{\Gamma}^{\alpha}_{\alpha \nu}+\overline{\Gamma}^{\alpha}_{\alpha \lambda}\overline{\Gamma}^{\lambda}_{\mu\nu}-\overline{\Gamma}^{\alpha}_{\mu\lambda}\overline{\lambda}^{\alpha}_{\alpha \nu}=0.
$$
With above preparations, we have the following important lemma, which are the represent formulas of the background FRW solutions.
\begin{lemma}\label{lemma:2.2}
There exists a family of solutions $(w(\tau),\overline{\Psi}(\tau))$ to \eqref{3.10}-\eqref{3.12}, and the solutions satisfy the following
\begin{equation}\label{3.13}
\aligned
\overline{\widetilde{h}} &=\frac{(IK)^{\frac{2\alpha}{\alpha+1}}\tau^{6\alpha}}{(1+K\tau^{3(1+\alpha)})^{\frac{2\alpha}{\alpha+1}}},\vspace{2mm}
\\
\widetilde{\del}^{\tau}\overline{\Psi} &=\frac{(IK)^{\frac{\alpha}{\alpha+1}}\tau^{3\alpha+1}w}{(1+K\tau^{3(1+\alpha)})^{\frac{\alpha}{\alpha+1}}},\vspace{2mm}\\
3w^{2} &=\frac{A^{\frac{2}{\alpha+1}}I^{\frac{1}{\alpha+1}}}{(I-\overline{\widetilde{h}}^{\frac{\alpha+1}{2\alpha}})^{\frac{1}{\alpha+1}}},\vspace{2mm}\\
\del_{t}w &=-\frac{A^{\frac{2}{\alpha+1}}I^{-\frac{\alpha}{\alpha+1}}\overline{\widetilde{h}}^{\frac{\alpha+1}{2\alpha}}}
{2(I-\overline{\widetilde{h}}^{\frac{\alpha+1}{2\alpha}})^{\frac{1}{\alpha+1}}},
\endaligned
\end{equation}
where $K$ is a constant depending only on the initial data.
\end{lemma}
\begin{proof}
 At first, we consider the $0$-th component of the fluid equations, we have by neglecting $A^{\frac{2}{\alpha+1}}I^{-\frac{\alpha}{\alpha+1}}$
\begin{eqnarray*}
\nabla_{\mu}\widetilde{T}^{\mu0}&=&\del_{\tau}\widetilde{T}^{00}+\del_{i}\widetilde{T}^{0i}+2\overline{\Gamma}^{0}_{00}\widetilde{T}^{00}\\
&=&\del_{\tau}\left(\frac{\tau^{2}w^{2}I}{(I-\overline{\widetilde{h}}^{\frac{\alpha+1}{2\alpha}})^{\frac{1}{\alpha+1}}}\right)-\frac{2\del_{\tau}w}{w}
\left(\frac{\tau^{2}w^{2}I}{(I-\overline{\widetilde{h}}^{\frac{\alpha+1}{2\alpha}})^{\frac{1}{\alpha+1}}}\right)\\
&=&\frac{6}{\tau}\left[\frac{\overline{\widetilde{h}}^{\frac{\alpha+1}{2\alpha}}\tau^{2}w^{2}}{(I-\overline{\widetilde{h}}
^{\frac{\alpha+1}{2\alpha}})^{\frac{1}{\alpha+1}}}
+(I-\overline{\widetilde{h}}^{\frac{\alpha+1}{2\alpha}})^{\frac{\alpha}{\alpha+1}}\tau^{2}w^{2}\right]\\
&&-\frac{w^{2}}{\tau}\left[\frac{\tau^{2}\overline{\widetilde{h}}^{\frac{\alpha+1}{\alpha}}}{(I-\overline{\widetilde{h}}^{\frac{\alpha+1}{\alpha}})^{\frac{1}{\alpha+1}}}
+4\tau^{2}(I-\overline{\widetilde{h}}^{\frac{\alpha+1}{2\alpha}})^{\frac{\alpha}{\alpha+1}}\right].
\end{eqnarray*}
Solving above ODE, we easily get
\begin{equation}\label{3.14}
I\tau\del_{\tau}\left(\frac{1}{(I-\overline{\widetilde{h}}^{\frac{\alpha+1}{2\alpha}})^{\frac{1}{\alpha+1}}}\right)
=\frac{3\overline{\widetilde{h}}^{\frac{\alpha+1}{2\alpha}}}{(I-\overline{\widetilde{h}}^{\frac{\alpha+1}{2\alpha}})^{\frac{1}{\alpha+1}}},
\end{equation}
\eqref{3.14} is equivalent to
\begin{equation*}
I\tau\del_{\tau}\overline{\widetilde{h}}=6\alpha\overline{\widetilde{h}}(I-\overline{\widetilde{h}}^{\frac{\alpha+1}{2\alpha}}).
\end{equation*}
By setting
$$
\xi=\overline{\widetilde{h}}^{\frac{\alpha+1}{2\alpha}}
$$
 we have
\begin{equation}\label{3.15}
\frac{I d\xi}{\xi(I-\xi)}=\frac{3(\alpha+1)d\tau}{\tau}.
\end{equation}
Solving \eqref{3.15}, we have
$$
\xi=\frac{IK\tau^{3(\alpha+1)}}{1+K\tau^{3(\alpha+1)}},
$$
where
$
K=\frac{\xi(1)}{I-\xi(1)}
$
and thus,
\begin{equation}\label{3.16}
\overline{\widetilde{h}}=\frac{(IK)^{\frac{2\alpha}{\alpha+1}}\tau^{6\alpha}}{(1+K\tau^{3(\alpha+1)})^{\frac{2\alpha}{\alpha+1}}}.
\end{equation}
Obviously, we have
\begin{equation}\label{3.17}
\widetilde{\del}^{\tau}\overline{\Psi}=(\tau^{2}w^{2}\overline{\widetilde{h}})^{\frac{1}{2}}
=\frac{(IK)^{\frac{\alpha}{\alpha+1}}\tau^{3\alpha+1}w}{(1+K\tau^{3(1+\alpha)})^{\frac{\alpha}{\alpha+1}}}.
\end{equation}
The $i$-th components of the fluid equations hold obviously.

 Now we consider the Einstein equations.

At first, from $\overline{R}_{ij}=0$, we have
\begin{equation}\label{3.18}
-\frac{w^{2}}{\tau^{2}}+\frac{w\del_{\tau}w}{\tau}-\frac{2w^{2}}{\tau^{2}}
+\frac{A^{\frac{2}{\alpha+1}}I^{-\frac{\alpha}{\alpha+1}}}{2\tau^{2}}\frac{2I-\overline{\widetilde{h}}^{\frac{\alpha+1}{2\alpha}}}
{(I-\overline{\widetilde{h}}^{\frac{\alpha+1}{2\alpha}})^{\frac{1}{\alpha+1}}}=0.
\end{equation}
Combining \eqref{3.18} with $\overline{R}_{00}=0$, we have
\begin{equation}\label{3.19}
-4(\frac{1}{\tau^{2}}-\frac{\del_{\tau}w}{\tau w})+\frac{4}{\tau^{2}}-\frac{2A^{\frac{2}{\alpha+1}}
I^{-\frac{\alpha}{\alpha+1}}\overline{\widetilde{h}}^{\frac{1-\alpha}{2\alpha}}\widetilde{\del}_{\tau}\overline{\Psi}\widetilde{\del}_{\tau}\overline{\Psi}}
{(I-\overline{\widetilde{h}}^{\frac{\alpha+1}{2\alpha}})^{\frac{1}{\alpha+1}}}=0.
\end{equation}
From \eqref{3.18}, \eqref{3.19} and $\del_{\tau}w=-\frac{\del_{t}w}{\tau w}$, we easily get
\begin{equation}\label{3.20}
\begin{cases}
3w^{2}-\frac{A^{\frac{2}{\alpha+1}}I^{\frac{1}{\alpha+1}}}{(I-\overline{\widetilde{h}}^{\frac{\alpha+1}{2\alpha}})^{\frac{1}{\alpha+1}}}=0,\vspace{2mm}\\
-\del_{t}w=\frac{A^{\frac{2}{\alpha+1}}I^{-\frac{\alpha}{\alpha+1}}\overline{\widetilde{h}}^{\frac{\alpha+1}{2\alpha}}}{2(I-\overline{\widetilde{h}}^{\frac{\alpha+1}{2\alpha}})^{\frac{1}{\alpha+1}}}.
\end{cases}
\end{equation}
The lemma holds by combing \eqref{3.16}, \eqref{3.17} and \eqref{3.20}.
\end{proof}
\begin{cor}\label{corollary:2.1}
By the above lemma and via a Taylor expansion, we see that when $\tau\rightarrow 0$,
\begin{equation}\label{3.21}
\begin{cases}
3w^{2}-A^{\frac{2}{\alpha+1}}\sim \tau^{3(\alpha+1)},\vspace{2mm}\\
-\del_{t}w\sim \tau^{3(\alpha+1)},\\
\del_{\tau}(\dot{w})\sim \tau^{3\alpha+2},\\
\del_{\tau}^{2}w\sim \tau^{3\alpha+1}.
\end{cases}
\end{equation}
\end{cor}

The asymptotic behavior above plays very important roles in the analysis of the source terms in Section \ref{section:3.4}.

\begin{proof}
The first and second asymptotic behavior can be obtained directly from \eqref{3.13}. The third asymptotic behavior can be derived by differentiating $\dot{w}$ directly and use \eqref{3.16}. For the last, we have
$$
\del_{\tau}^{2}w=\del_{\tau}(-\frac{\dot{w}}{\tau w})=-\frac{\del_{\tau}\dot{w}}{\tau w}+
\frac{\dot{w}w}{(\tau w)^{2}}-\frac{(\dot{w})^{2}}{\tau^{2}w^{3}}\sim\tau^{3\alpha+1}.
$$
\end{proof}

\begin{remark}\label{remark:2.3}
Since we have confined $\tau\in(0,1]$, which means that we set initially $a(0)=1$. Thus, we have to prove that when $t\rightarrow \infty$, $a(t)\rightarrow +\infty$. From \eqref{3.13}, there must exist two positive constants $A_{1}$ and $A_{2}$ such that
$$
A_{1}\leq w(t)=\frac{\dot{a}(t)}{a(t)}\leq A_{2},
$$
then by comparison theorem of ODE,
$
e^{A_{1}t}\leq a(t)\leq e^{A_{2}t}.
$
It is obvious that $a(t)=+\infty$, when $t\rightarrow +\infty$.
\end{remark}

%==============================================================

\section{Proof of the nonlinear stability property}\label{section:3}

\subsection{The wave gauge}\label{section:3.1}

Our first task is to introdude appropriate coordinates.
For the background metric $\widetilde{\eta}$, we have
\begin{equation*}
\overline{\widetilde{\Gamma}}^{0}_{00}=-\frac{1}{\tau}+\frac{\dot{w}}{\tau w^{2}},
\end{equation*}
and
\begin{equation*}
\overline{\widetilde{\Gamma}}^{0}_{ii}=-\frac{1}{2}\widetilde{\eta}^{00}\del_{\tau}\widetilde{\eta}_{ii}=-\frac{w^{2}}{\tau}.
\end{equation*}
Then, we obtain
\begin{equation*}
\overline{\widetilde{\Gamma}}^{0}=\widetilde{\eta}^{00}\overline{\widetilde{\Gamma}}^{0}_{00}+\widetilde{\eta}^{ii}\overline{\widetilde{\Gamma}}^{0}_{ii}=-2\tau w^{2}-\tau\dot{w},
\end{equation*}
and
\begin{equation}\label{4.1}
e^{2\Phi}\overline{\widetilde{\Gamma}}^{0}=-\frac{2w^{2}}{\tau}-\frac{\dot{w}}{\tau}.
\end{equation}
Direct calculations show that
\begin{equation}\label{4.2}
\overline{\widetilde{\Gamma}}^{i}=0,\quad(i=1,2,3).
\end{equation}
For the conformal metric $\eta$, we have
\begin{equation}\label{4.3}
\overline{\Gamma}^{\mu}=-\frac{\dot{w}}{\tau}\delta^{\mu}_{0}.
\end{equation}
On the other hand, under conformal transformation \eqref{3.8}
\begin{equation}\label{4.4}
\Gamma^{\mu}=g^{\alpha\beta}\Gamma_{\alpha\beta}^{\mu}=2\nabla^{\mu}\Phi+e^{2\Phi}\widetilde{\Gamma}^{\mu}.
\end{equation}
Define the wave coordinates as
\begin{equation}\label{4.5}
Z^{\mu}=\Gamma^{\mu}+Y^{\mu}=\Gamma^{\mu}+\frac{2}{\tau}\left(g^{\mu0}+(w^{2}+\frac{\dot{w}}{2})
\delta^{\mu}_{0}\right):=\Gamma^{\mu}+\frac{2}{\tau}\left(g^{\mu0}+\frac{\Lambda(\tau)}{3}\delta^{\mu}_{0}\right).
\end{equation}
In \eqref{4.5}, we have denoted $w^{2}+\frac{\dot{w}}{2}$ by $\frac{\Lambda(\tau)}{3}$ for convenience.

\begin{remark}\label{remark:3.1}
1. By results in Zengino\u{g}lu \cite{Z}, we see that if initially $Z^{\mu}=0$, then $Z^{\mu}\equiv0$ in the whole evolution of the Einstein-Euler system.

2.
From \eqref{4.5}, we can see that $Y^{\mu}=-2\nabla^{\mu}\Phi-e^{2\Phi}\overline{\widetilde{\Gamma}}^{\mu}$. From \eqref{4.1}-\eqref{4.4} we see that for the background metric $\eta$, $Z^{\mu}\equiv 0$. This fact is very important for the disappear of the linear parts of the conformal Einstein and conformal fluid equations \eqref{1.11}.

3.
When $w^{2}=\frac{\Lambda}{3}$ with $\Lambda$ a positive constant, this is exactly the case considered by Oliynyk \cite{O}.
\end{remark}

%----------------------------------------------------------------------------------------------------------------------

\subsection{The reduced conformal Einstein equations}\label{section:3.2}

With the wave coordinates $Z^{\mu}$ defined by \eqref{4.5}, we can consider the following equivalently reduced conformal Einstein equations by assuming $Z^{\mu}|_{\tau=1}=0$
\begin{eqnarray}\label{4.6}
-2R^{\mu\nu}&+&2\nabla^{(\mu}Z^{\nu)}+A_{\kappa}^{\mu\nu}Z^{\kappa}=-4\nabla^{\mu}\nabla^{\nu}\Phi+4\nabla^{\mu}\Phi\nabla^{\nu}\Phi\nonumber\\
&-&2\left[\Box_{g}\Phi+2|\nabla\Phi|^{2}_{g}+\frac{1}{2}A^{\frac{2}{\alpha+1}}I^{-\frac{\alpha}{\alpha+1}}\frac{2I-\widetilde{h}^{\frac{\alpha+1}{2\alpha}}}
{(I-\widetilde{h}^{\frac{\alpha+1}{2\alpha}})^{\frac{1}{\alpha+1}}}e^{2\Phi}\right]g^{\mu\nu}\nonumber\\
&-&2e^{4\Phi}\frac{A^{\frac{2}{\alpha+1}}I^{-\frac{\alpha}{\alpha+1}}\widetilde{h}^{\frac{1-\alpha}{2\alpha}}\widetilde{\del}^{\mu}\Psi\widetilde{\del}^{\nu}\Psi}
{(I-\widetilde{h}^{\frac{\alpha+1}{2\alpha}})^{\frac{1}{\alpha+1}}}.
\end{eqnarray}
With
$$
A^{\mu\nu}_{\kappa}=-\Gamma^{(\mu}\delta^{\nu)}_{\kappa}+Y^{(\mu}\delta^{\nu)}_{\kappa}\quad\text{and}\quad
\nabla^{(\mu}Z^{\nu)}=\frac{1}{2}(\nabla^{\mu}Z^{\nu}+\nabla^{\nu}Z^{\mu}),
$$
we have
\begin{eqnarray*}
2\nabla^{(\mu}Z^{\nu)}&=&2\nabla^{(\mu}\Gamma^{\nu)}+\nabla^{\mu}Y^{\nu}+\nabla^{\nu}Y^{\mu}\nonumber\\
&=&2\nabla^{(\mu}\Gamma^{\nu)}+\del_{\tau}\left(\frac{2\Lambda(\tau)}{3\tau}\right)(g^{i0}\delta^{\nu}_{0}+g^{\nu0}\delta^{\mu}_{0})
-\frac{2\Lambda(\tau)}{3\tau}
\del_{\tau}g^{\mu\nu}-4\nabla^{\mu}\nabla^{\nu}\Phi
\end{eqnarray*}
and
\begin{eqnarray*}
A_{k}^{\mu\nu}Z^{\kappa}&=&(-\Gamma^{(\mu}\delta^{\nu)}_{\kappa}+Y^{(\mu}\delta^{\nu)}_{\kappa})(\Gamma^{\kappa}+Y^{\kappa})\nonumber\\
&=&-\Gamma^{(\mu}\Gamma^{\nu)}+4\nabla^{\mu}\Phi\nabla^{\nu}\Phi-\frac{4\Lambda(\tau)}{3\tau}(\nabla^{\mu}\Phi\delta^{\nu}_{0}+\nabla^{\nu}\Phi\delta^{\mu}_{0})
+\frac{4\Lambda^{2}(\tau)}{9\tau^{2}}\delta^{\mu}_{0}\delta^{\nu}_{0}.
\end{eqnarray*}
The Einstein equations \eqref{3.11} become
\begin{eqnarray*}
-2R^{\mu\nu}&+&2\nabla^{(\mu}\Gamma^{\nu)}-\Gamma^{\mu}\Gamma^{\nu}=\frac{2\Lambda(\tau)}{3\tau}\del_{\tau}g^{\mu\nu}
-\frac{2\del_{\tau}\Lambda(\tau)}{3\tau}
(g^{\mu0}\delta^{\nu}_{0}+g^{\nu0}\delta^{\mu}_{0})\nonumber\\
&-&\frac{4\Lambda(\tau)}{3\tau^{2}}\left(g^{00}+\frac{\Lambda(\tau)}{3}\right)\delta^{\mu}_{0}\delta^{\nu}_{0}-
\frac{4\Lambda(\tau)}{3\tau^{2}}g^{0i}\delta_{0}^{(\mu}\delta^{\nu)}_{i}-\frac{2}{\tau^{2}}g^{\mu\nu}\left(g^{00}+\frac{\Lambda(\tau)}{3}\right)\nonumber\\
&+&\frac{2}{\tau^{2}}\left(\Lambda(\tau)-A^{\frac{2}{\alpha+1}}I^{-\frac{\alpha}{\alpha+1}}\frac{I-\frac{1}{2}
\widetilde{h}^{\frac{\alpha+1}{2\alpha}}}
{(I-\widetilde{h}^{\frac{\alpha+1}{2\alpha}})^{\frac{1}{\alpha+1}}}\right)g^{\mu\nu}\nonumber\\
&-&2e^{4\Phi}\frac{A^{\frac{2}{\alpha+1}}I^{-\frac{\alpha}{\alpha+1}}\widetilde{h}^{\frac{1-\alpha}{2\alpha}}\widetilde{\del}^{\mu}\Psi\widetilde{\del}^{\nu}\Psi}
{(I-\widetilde{h}^{\frac{\alpha+1}{2\alpha}})^{\frac{1}{\alpha+1}}}.
\end{eqnarray*}
Expanding the left hand of above and inserting $\frac{\Lambda(\tau)}{3}=w^{2}+\frac{\dot{w}}{2}$, we get
\begin{eqnarray*}
-g^{\kappa\lambda}\del_{\kappa}\del_{\lambda}g^{\mu\nu}&=&\frac{2w^{2}}{\tau}\del_{\tau}g^{\mu\nu}-\frac{4w^{2}}{\tau^{2}}(g^{00}+w^{2})
\delta^{\mu}_{0}\delta^{\nu}_{0}
-\frac{4w^{2}}{\tau^{2}}g^{0i}\delta_{0}^{(\mu}\delta^{\nu)}_{i}
\nonumber\\
&-&\frac{2}{\tau^{2}}g^{\mu\nu}(g^{00}+w^{2})
-\frac{2\dot{w}}{\tau^{2}}(g^{00}+w^{2}+\frac{\dot{w}}{2})\delta^{\mu}_{0}\delta^{\nu}_{0}
-\frac{\dot{w}}{\tau}\del_{\tau}g^{\mu\nu}
\nonumber\\
&-&\frac{2}{\tau^{2}}w^{2}\dot{w}\delta^{\mu}_{0}\delta^{\nu}_{0}
-\frac{2\dot{w}}{\tau^{2}}g^{0i}\delta_{0}^{(\mu}\delta^{\nu)}_{i}-\frac{\dot{w}}{\tau^{2}}g^{\mu\nu}
-\frac{2\del_{\tau}\Lambda(\tau)}{3\tau}(g^{\mu0}\delta^{\nu}_{0}+g^{\nu0}\delta^{\mu}_{0})\nonumber\\
&+&\frac{2}{\tau^{2}}\left(3w^{2}+\frac{3\dot{w}}{2}-
A^{\frac{2}{\alpha+1}}I^{-\frac{\alpha}{\alpha+1}}\frac{I-\frac{1}{2}\widetilde{h}^{\frac{\alpha+1}{2\alpha}}}
{(I-\widetilde{h}^{\frac{\alpha+1}{2\alpha}})^{\frac{1}{\alpha+1}}}\right)g^{\mu\nu}\nonumber\\
&-&2e^{4\Phi}\frac{A^{\frac{2}{\alpha+1}}I^{-\frac{\alpha}{\alpha+1}}\widetilde{h}^{\frac{1-\alpha}{2\alpha}}\widetilde{\del}^{\mu}\Psi\widetilde{\del}^{\nu}\Psi}
{(I-\widetilde{h}^{\frac{\alpha+1}{2\alpha}})^{\frac{1}{\alpha+1}}}+Q^{\mu\nu}(g,\del g)\\
&:=&\frac{2w^{2}}{\tau}\del_{\tau}g^{\mu\nu}-\frac{4w^{2}}{\tau^{2}}(g^{00}+w^{2})\delta^{\mu}_{0}\delta^{\nu}_{0}
-\frac{4w^{2}}{\tau^{2}}g^{0i}\delta_{0}^{(\mu}\delta^{\nu)}_{i}\nonumber\\
&-&\frac{2}{\tau^{2}}g^{\mu\nu}(g^{00}+w^{2})+M^{\mu\nu}
\nonumber
\end{eqnarray*}
or equivalently
\begin{eqnarray}\label{4.7}
-g^{\kappa\lambda}\del_{\kappa}\del_{\lambda}(g^{\mu\nu}-\eta^{\mu\nu})
&=&\frac{2w^{2}}{\tau}\del_{\tau}(g^{\mu\nu}-\eta^{\mu\nu})-\frac{4w^{2}}{\tau^{2}}(g^{00}+w^{2})\delta^{\mu}_{0}\delta^{\nu}_{0}\nonumber\\
&&-\frac{4w^{2}}{\tau^{2}}g^{0i}\delta_{0}^{(\mu}\delta^{\nu)}_{i}-\frac{2}{\tau^{2}}g^{\mu\nu}(g^{00}+w^{2})+\hat{M}^{\mu\nu}.
\end{eqnarray}
We have
\begin{eqnarray}\label{4.8}
\hat{M}^{\mu\nu}&=&M^{\mu\nu}+g^{\kappa\lambda}\del_{\kappa}\del_{\lambda}\eta^{\mu\nu}+2\frac{w^{2}}{\tau}\del_{\tau}\eta^{\mu\nu}\nonumber\\ 
&=&(g^{\kappa\lambda}-\eta^{\kappa\lambda})\del_{\kappa}\del_{\lambda}\eta^{\mu\nu}-\frac{2\dot{w}}{\tau^{2}}(g^{00}+w^{2})\delta^{\mu}_{0}\delta^{\nu}_{0}
-\frac{\dot{w}}{\tau}\del_{\tau}(g^{\mu\nu}-\eta^{\mu\nu})-\frac{2\dot{w}}{\tau^{2}}g^{0i}\delta_{0}^{(\mu}\delta^{\nu)}_{i}\nonumber\\
&&-\frac{\dot{w}}{\tau^{2}}(g^{\mu\nu}-\eta^{\mu\nu})-\frac{2\del_{\tau}\Lambda(\tau)}{3\tau}\left((g^{\mu0}-\eta^{\mu0})\delta^{\nu}_{0}+(g^{\nu0}-\eta^{\nu0})\delta^{\mu}_{0}\right)\nonumber\\
&&+\frac{2}{\tau^{2}}A^{\frac{2}{\alpha+1}}I^{-\frac{\alpha}{\alpha+1}}\left(\frac{I-\frac{1}{2}\widetilde{h}^{\frac{\alpha+1}{2\alpha}}}
{(I-\widetilde{h}^{\frac{\alpha+1}{2\alpha}})^{\frac{1}{\alpha+1}}}
-\frac{I-\frac{1}{2}\overline{\widetilde{h}}^{\frac{\alpha+1}{2\alpha}}}{(I-\overline{\widetilde{h}}^{\frac{\alpha+1}{2\alpha}})^{\frac{1}{\alpha+1}}}\right)g^{\mu\nu}\nonumber\\
&&-2e^{4\Phi}A^{\frac{2}{\alpha+1}}I^{-\frac{\alpha}{\alpha+1}}\left(\frac{\widetilde{h}^{\frac{1-\alpha}{2\alpha}}\widetilde{\del}^{\mu}\Psi\widetilde{\del}^{\nu}\Psi}
{(I-\widetilde{h}^{\frac{\alpha+1}{2\alpha}})^{\frac{1}{\alpha+1}}}-\frac{\overline{\widetilde{h}}^{\frac{1-\alpha}{2\alpha}}\widetilde{\del}^{\mu}\overline{\Psi}\widetilde{\del}^{\nu}\overline{\Psi}}
{(I-\overline{\widetilde{h}}^{\frac{\alpha+1}{2\alpha}})^{\frac{1}{\alpha+1}}}\right)\nonumber\\
&&+Q^{\mu\nu}(g,\del g)-Q^{\mu\nu}(\eta,\del\eta)
\end{eqnarray}

\begin{remark}\label{remark:3.2}
1. In above, $Q^{\mu\nu}(g,\del g)$ are quadratic in $\del g=(\del_{\kappa}g^{\mu\nu})$ and analytical in $g=(g^{\mu\nu})$.
Thus, in $\hat{M}^{\mu\nu}$, each term contains $g-\eta$ or $\widetilde{\del}\Psi-\widetilde{\del}\overline{\Psi}$, which will be proved later.

2. Equation \eqref{3.11} is equivalent to \eqref{4.7}, provided that $Z^{\mu}=0$. On the other hand, the background solution $(\eta,\overline{\Psi}(\tau))$ defined in Section \ref{section:2.3} satisfies \eqref{4.7} obviously, since $Z^{\mu}\equiv0$ for the metric $\eta$, which is also the reason for the disappear of the linear parts of \eqref{4.8}
\end{remark}
The main part of this paper is turning the Einstein-Euler system of GCG into a symmetric hyperbolic system, thus we make this process clear in the following.

Recall the densitized three metric defined in Section \ref{section:1.3}
\begin{equation}\label{4.9}
\textbf{g}^{ij}=det(\check{g}_{lm})^{\frac{1}{3}}g^{ij},
\end{equation}
where
$$
\check{g}_{lm}=(g^{lm})^{-1},
$$
and the variable
\begin{equation}\label{4.10}
\textbf{q}=g^{00}+w^{2}-\frac{w^{2}}{3}\ln(\det(g^{pq})).
\end{equation}
It is easy to check that
\begin{equation}\label{4.11}
\del_{\mu}\textbf{g}^{ij}=(\det(\check{g}_{pq}))^{\frac{1}{3}}\textbf{L}_{lm}^{ij}\del_{\mu}g^{lm},
\end{equation}
where
$$
\textbf{L}_{lm}^{ij}=\delta^{i}_{l}\delta^{j}_{m}-\frac{1}{3}\check{g}_{lm}g^{ij}.
$$
Obviously, $\textbf{L}_{lm}^{ij}$ is trace-free, i.e.,
$$
\textbf{L}_{lm}^{ij}g^{lm}=0.
$$
Define
\begin{eqnarray}\label{4.12}
\textbf{u}^{0\nu}&=&\frac{g^{0\nu}-\eta^{0\nu}}{2\tau},\\\label{4.13}
\textbf{u}_{0}^{0\nu}&=&\del_{\tau}(g^{0\nu}-\eta^{0\nu})-\frac{3(g^{0\nu}-\eta^{0\nu})}{2\tau},\\\label{4.14}
\textbf{u}_{i}^{0\nu}&=&\del_{i}(g^{0\nu}-\eta^{0\nu}),\\\label{4.15}
\textbf{u}^{ij}&=&\textbf{g}^{ij}-\delta^{ij},\\\label{4.16}
\textbf{u}_{\mu}^{ij}&=&\del_{\mu}\textbf{g}^{ij},\\\label{4.17}
\textbf{u}&=&\textbf{q},\\\label{4.18}
\textbf{u}_{\mu}&=&\del_{\mu}\textbf{q}.
\end{eqnarray} 
At first, we consider the equation satisfied by $g^{0\mu}-\eta^{0\mu}$.

 Utilizing \eqref{4.12}-\eqref{4.14} and inserting $g^{0\mu}-\eta^{0\mu}=2\tau \textbf{u}^{0\mu}$,
 $\del_{i}(g^{0\mu}-\eta^{0\mu})=\textbf{u}_{i}^{0\mu}$ and $\del_{\tau}(g^{0\mu}-\eta^{0\mu})=\textbf{u}_{0}^{0\mu}+3\textbf{u}^{0\mu}$ into \eqref{4.7}, we have
\begin{eqnarray*}
&&-g^{00}\del_{\tau}(\textbf{u}_{0}^{0\mu}+3\textbf{u}^{0\mu})-2g^{0i}\del_{i}(\textbf{u}_{0}^{0\mu}+3\textbf{u}^{0\mu})-g^{ij}\del_{j}\textbf{u}_{i}^{0\mu}\\
&=&\frac{2w^{2}}{\tau}(\textbf{u}_{0}^{0\mu}+3\textbf{u}^{0\mu})-\frac{4w^{2}}{\tau^{2}}(2\tau \textbf{u}^{00})\delta^{\mu}_{0}-\frac{4w^{2}}{\tau^{2}}(2\tau \textbf{u}^{0i})
\delta_{0}^{(\mu}\delta^{0)}_{i}-\frac{4}{\tau}g^{0\mu}\textbf{u}^{00}+\hat{M}^{0\mu}.
\end{eqnarray*}
Based on above, we get
\begin{eqnarray}\label{4.19}
&&-g^{00}\del_{\tau} \textbf{u}_{0}^{0\mu}-2g^{0i}\del_{i}\textbf{u}_{0}^{0\mu}-g^{ij}\del_{j}\textbf{u}_{i}^{0\mu}\nonumber\\
&=&3g^{00}\del_{\tau}\left(\frac{g^{0\mu}-\eta^{0\mu}}{2\tau}\right)+6g^{0i}\del_{i}\left(\frac{g^{0\mu}-\eta^{0\mu}}{2\tau}\right)
+\frac{2}{\tau}(2\tau \textbf{u}^{00}-g^{00})(\textbf{u}_{0}^{0\mu}+3\textbf{u}^{0\mu})\nonumber\\
&&-\frac{8}{\tau}(2\tau \textbf{u}^{00}-g^{00})\textbf{u}^{00}\delta_{0}^{\mu}-\frac{8}{\tau}(2\tau \textbf{u}^{00}-g^{00})\textbf{u}^{0i}\delta_{0}^{(\mu}\delta^{0)}_{i}
-\frac{4}{\tau}g^{0\mu}\textbf{u}^{00}+\hat{M}^{0\mu}\nonumber\\
&=&3g^{00}\frac{\textbf{u}_{0}^{0\mu}+3\textbf{u}^{0\mu}}{2\tau}-\frac{3g^{00}\textbf{u}^{0\mu}}{\tau}+6\textbf{u}^{0i}\textbf{u}_{i}^{0\mu}
+4\textbf{u}^{00}\textbf{u}_{0}^{0\mu}+12\textbf{u}^{00}\textbf{u}^{0\mu}\nonumber\\
&&-\frac{2}{\tau}g^{00}\textbf{u}_{0}^{0\mu}-\frac{6}{\tau}g^{00}\textbf{u}^{0\mu}-16\textbf{u}^{00}\textbf{u}^{00}\delta^{\mu}_{0}
+\frac{8}{\tau}g^{00}\textbf{u}^{00}\delta^{\mu}_{0}
-16\textbf{u}^{00}\textbf{u}^{0i}\delta_{0}^{(\mu}\delta^{0)}_{i}\nonumber\\
&&+\frac{8g^{00}}{\tau}\textbf{u}^{0i}\delta_{0}^{(\mu}\delta^{0)}_{i}-\frac{4}{\tau}g^{0\mu}\textbf{u}^{00}+\hat{M}^{0\mu}\nonumber\\
&=&\left\{\begin{array}{ll}
\frac{1}{\tau}[-\frac{g^{00}}{2}(\textbf{u}_{0}^{00}+\textbf{u}^{00})]+6\textbf{u}^{0i}\textbf{u}^{00}
+4\textbf{u}^{00}\textbf{u}_{0}^{00}-4\textbf{u}^{00}\textbf{u}^{00}+\hat{M}^{00},\quad \mu=0\\\vspace{2mm}
\frac{1}{\tau}[-\frac{g^{00}}{2}(\textbf{u}_{0}^{0k}+\textbf{u}^{0k})]+6\textbf{u}^{0i}\textbf{u}^{0k}
+4\textbf{u}^{00}\textbf{u}_{0}^{0k}-4\textbf{u}^{00}\textbf{u}^{0k}+\hat{M}^{0k},\quad \mu=k
\end{array}
\right.\nonumber\\
&=&\frac{1}{\tau}\left[-\frac{g^{00}}{2}(\textbf{u}_{0}^{0\mu}+\textbf{u}^{0\mu})\right]+6\textbf{u}^{0i}\textbf{u}^{0\mu}+4\textbf{u}^{00}\textbf{u}_{0}^{0\mu}
-4\textbf{u}^{00}\textbf{u}^{0\mu}+\hat{M}^{0\mu}.
\end{eqnarray}
Now we consider $\textbf{g}^{ij}-\delta^{ij}$, direct calculations give
\begin{eqnarray}\label{4.20}
-g^{\kappa\lambda}\del_{\kappa}\del_{\lambda}\textbf{g}^{ij}&=&-g^{\kappa\lambda}\del_{\kappa}[(det(\check{g}_{pq}))^{\frac{1}{3}}\textbf{L}_{lm}^{ij}\del_{\lambda}g^{lm}]\nonumber\\
&=&(det(\check{g}_{pq}))^{\frac{1}{3}}\textbf{L}_{lm}^{ij}(-g^{\kappa\lambda}\del_{\kappa}\del_{\lambda}g^{lm})-g^{\kappa\lambda}\del_{\kappa}[(det(\check{g}_{pq}))^{\frac{1}{3}}
\textbf{L}_{lm}^{ij}]\del_{\lambda}g^{lm}\nonumber\\
&=&(det(\check{g}_{pq}))^{\frac{1}{3}}\textbf{L}_{lm}^{ij}\left(\frac{2w^{2}}{\tau}\del{\tau}(g^{lm}-\eta^{lm})
-\frac{2}{\tau^{2}}g^{lm}(g^{00}+w^{2})+\hat{M}^{lm}\right)\nonumber\\
&&-g^{\kappa\lambda}\del_{\kappa}[(det(\check{g}_{pq}))^{\frac{1}{3}}
\textbf{L}_{lm}^{ij}]\del_{\lambda}g^{lm}\nonumber\\
&=&\frac{2w^{2}}{\tau}\del_{\tau}\textbf{g}^{ij}+(det(\check{g}_{pq}))^{\frac{1}{3}}\textbf{L}_{lm}^{ij}\hat{M}^{lm}-
g^{\kappa\lambda}\del_{\kappa}[(det(\check{g}_{pq}))^{\frac{1}{3}}
\textbf{L}_{lm}^{ij}]\del_{\lambda}g^{lm}\nonumber\\
&:=&\frac{2w^{2}}{\tau}\del_{\tau}\textbf{g}^{ij}+\acute{M}^{ij}.
\end{eqnarray}
We have
$$
\acute{M}^{ij}=(det(\check{g}_{pq}))^{\frac{1}{3}}\textbf{L}_{lm}^{ij}\hat{M}^{lm}-
g^{\kappa\lambda}\del_{\kappa}[(det(\check{g}_{pq}))^{\frac{1}{3}}
\textbf{L}_{lm}^{ij}]\del_{\lambda}g^{lm}.
$$
Then by \eqref{4.15}-\eqref{4.16}, we have
\begin{eqnarray}\label{4.21}
&&-g^{00}\del_{\tau}\textbf{u}_{0}^{ij}-2g^{0i}\del_{i}\textbf{u}_{0}^{ij}-g^{pq}\del_{p}\textbf{u}_{q}^{ij}\nonumber\\
&=&\frac{2}{\tau}(2\tau \textbf{u}^{00}-g^{00})\textbf{u}_{0}^{ij}+\acute{M}^{ij}\nonumber\\
&=&-\frac{2}{\tau}g^{00}\textbf{u}_{0}^{ij}+4\textbf{u}^{00}\textbf{u}_{0}^{ij}+\acute{M}^{ij}.
\end{eqnarray}
At last, we consider $\textbf{q}$, we have by \eqref{4.10}
\begin{equation}\label{4.22}
\del_{\lambda}\textbf{q}=\del_{\lambda}(g^{00}+w^{2})
-\frac{w^{2}}{3}g_{pq}\del_{\lambda}g^{pq}-\frac{2w\del_{\tau}w\delta^{\lambda}_{0}}{3}\ln(\det(g^{pq})).
\end{equation}
Then by \eqref{4.22}
\begin{eqnarray}\label{4.23}
\del_{\kappa}\del_{\lambda}\textbf{q}&=&\del_{\kappa}\del_{\lambda}(g^{00}+w^{2})-\frac{w^{2}}{3}g_{pq}\del_{\kappa}\del_{\lambda}g^{pq}-\frac{w^{2}}{3}\del_{\kappa}g_{pq}\del_{\lambda}g^{pq}\nonumber\\
&&-\frac{2w\del_{\tau}w\delta^{\kappa}_{0}}{3}g_{pq}\del_{\lambda}g^{pq}-\frac{2(\del_{\tau}w)^{2}}{3}\delta^{\lambda}_{0}\delta^{\kappa}_{0}\ln(\det(g^{pq}))\nonumber\\
&&-\frac{2w\del^{2}_{\tau}w}{3}\delta^{\lambda}_{0}\delta^{\kappa}_{0}\ln(\det(g^{pq}))-\frac{2w\del_{\tau}w}{3}\delta^{\lambda}_{0}g_{pq}\del_{\kappa}g^{pq}\nonumber\\
&:=&\del_{\kappa}\del_{\lambda}(g^{00}+w^{2})-\frac{w^{2}}{3}g_{pq}\del_{\kappa}\del_{\lambda}g^{pq}+R^{\textbf{q}},
\end{eqnarray}
where
\begin{eqnarray*}
R^{\textbf{q}}&=&-\frac{w^{2}}{3}\del_{\kappa}g_{pq}\del_{\lambda}g^{pq}
-\frac{2w\del_{\tau}w\delta^{\kappa}_{0}}{3}g_{pq}\del_{\lambda}g^{pq}-\frac{2(\del_{\tau}w)^{2}}{3}\delta^{\lambda}_{0}\delta^{\kappa}_{0}\ln(\det(g^{pq}))\nonumber\\
&&-\frac{2w\del^{2}_{\tau}w}{3}\delta^{\lambda}_{0}\delta^{\kappa}_{0}\ln(\det(g^{pq}))-\frac{2w\del_{\tau}w}{3}\delta^{\lambda}_{0}g_{pq}\del_{\kappa}g^{pq}.\nonumber\\
\end{eqnarray*}
Thus, we have by \eqref{4.17}-\eqref{4.18}
\begin{eqnarray}\label{4.24}
-g^{\kappa\lambda}\del_{\kappa}\del_{\lambda}\textbf{q}&=&-g^{\kappa\lambda}\del_{\kappa}\del_{\lambda}(g^{00}+w^{2})
+\frac{w^{2}}{3}g_{pq}g^{\kappa\lambda}
\del_{\kappa}\del_{\lambda}g^{pq}-g^{\kappa\lambda}R^{\textbf{q}}\nonumber\\
&=&\frac{2w^{2}}{\tau}\del_{\tau}(g^{00}+w^{2})-\frac{4w^{2}}{\tau^{2}}(g^{00}+w^{2})-\frac{2}{\tau^{2}}g^{00}(g^{00}+w^{2})+\hat{M}^{00}\nonumber\\
&&-\frac{w^{2}}{3}g_{pq}(\frac{2w^{2}}{\tau}\del_{\tau}g^{pq}-\frac{2}{\tau^{2}}g^{pq}(g^{00}+w^{2})+\hat{M}^{pq})-g^{\kappa\lambda}R^{\textbf{q}}\nonumber\\
&=&\frac{2w^{2}}{\tau}\del_{\tau}\textbf{q}-2\left(\frac{g^{00}+w^{2}}{\tau}\right)^{2}+\hat{M}^{00}-\frac{w^{2}}{3}g_{pq}\hat{M}^{pq}\nonumber\\
&&+\frac{4w^{3}\del_{\tau}w}{3\tau}\ln(\det{g^{pq}})-g^{\kappa\lambda}R^{\textbf{q}}\nonumber\\
&:=&\frac{2}{\tau}(2\tau \textbf{u}^{00}-g^{00})\del_{\tau}\textbf{q}-8(\textbf{u}^{00})^{2}+\hat{R}^{\textbf{q}}\nonumber\\
&=&-\frac{2}{\tau}g^{00}\del_{\tau}\textbf{q}+4\textbf{u}^{00}\del_{\tau}\textbf{q}-8(\textbf{u}^{00})^{2}+\hat{R}^{\textbf{q}}.
\end{eqnarray}
We have
$$
\hat{R}^{\textbf{q}}=\hat{M}^{00}-\frac{w^{2}}{3}g_{pq}\hat{M}^{pq}
+\frac{4w^{3}\del_{\tau}w}{3\tau}\ln(\det(g^{pq}))-g^{\kappa\lambda}R^{\textbf{q}}.
$$
From \eqref{4.24}, we have
\begin{eqnarray}\label{4.25}
&&-g^{00}\del_{\tau}\textbf{u}_{0}-2g^{0i}\del_{i}\textbf{u}_{0}-g^{ij}\del_{i}\textbf{u}_{j}\nonumber\\
&=&-\frac{2}{\tau}g^{00}\textbf{u}_{0}+4\textbf{u}^{00}\textbf{u}_{0}-8(\textbf{u}^{00})^{2}+\hat{R}^{\textbf{q}}.
\end{eqnarray}
From \eqref{4.19}, \eqref{4.21} and \eqref{4.25}, we easily transform the Einstein equations into the following symmetric hyperbolic system
\begin{eqnarray}\label{4.26}
A^{\kappa}\del_{\kappa}\left(\begin{array}{ll}
\textbf{u}^{0\mu}_{0}\\\textbf{u}^{0\mu}_{j}\\\textbf{u}^{0\mu}
\end{array}\right)=\frac{1}{\tau}\textbf{AP}\left(\begin{array}{ll}
\textbf{u}^{0\mu}_{0}\\\textbf{u}^{0\mu}_{j}\\\textbf{u}^{0\mu}
\end{array}\right)+F^{0\mu},
\end{eqnarray}
\begin{eqnarray}\label{4.27}
A^{\kappa}\del_{\kappa}\left(\begin{array}{ll}
\textbf{u}^{lm}_{0}\\ \textbf{u}^{lm}_{j}\\ \textbf{u}^{lm}
\end{array}\right)=\frac{1}{\tau}(-2g^{00})\Pi\left(\begin{array}{ll}
\textbf{u}^{lm}_{0}\\ \textbf{u}^{lm}_{j}\\ \textbf{u}^{lm}
\end{array}\right)+F^{lm},
\end{eqnarray}
and
\begin{eqnarray}\label{4.28}
A^{\kappa}\del_{\kappa}\left(\begin{array}{ll}
\textbf{u}_{0}\\ \textbf{u}_{j}\\ \textbf{u}
\end{array}\right)=\frac{1}{\tau}(-2g^{00})\Pi\left(\begin{array}{ll}
\textbf{u}_{0}\\ \textbf{u}_{j}\\ \textbf{u}
\end{array}\right)+F^{\textbf{q}},
\end{eqnarray}
where
\begin{align*}
A^{0}=\left(
\begin{array}{ccc}
-g^{00}& 0& 0\\
0& g^{ij}& 0\\
0&0& -g^{00}
\end{array}\right),\quad
A^{k}=\left(
\begin{array}{ccc}
-2g^{0k}& -g^{jk}& 0\\
-g^{ik}& 0& 0\\
0&0& 0
\end{array}\right),
\end{align*}
\begin{align*}
\textbf{P}=\left(
\begin{array}{ccc}
\frac{1}{2}& 0& \frac{1}{2}\\
0& \delta^{j}_{k}& 0\\
\frac{1}{2}&0& \frac{1}{2}
\end{array}\right),\quad
\textbf{A}=\left(
\begin{array}{ccc}
-g^{00}& 0& 0\\
0& \frac{3}{2}g^{jk}& 0\\
0&0& -g^{00}
\end{array}\right),
\end{align*}
\begin{align*}
\Pi=\left(
\begin{array}{ccc}
1& 0& 0\\
0& 0& 0\\
0&0& 0\end{array}\right),\quad
F^{0\mu}=\left(
\begin{array}{ccc}
&6\textbf{u}^{0i}\textbf{u}^{0\mu}+4\textbf{u}^{00}\textbf{u}_{0}^{0\mu}-4\textbf{u}^{00}\textbf{u}^{0\mu}+\hat{M}^{0\mu}\\
&0\\
&0
\end{array}\right),
\end{align*}
and
\begin{align*}
F^{ij}=\left(
\begin{array}{ccc}
&4\textbf{u}^{00}\textbf{u}_{0}^{ij}+\acute{M}^{ij}\\
&0\\
&\-g^{00}\textbf{u}_{0}^{lm}
\end{array}
\right),\quad
F^{\textbf{q}}=\left(
\begin{array}{ccc}
&4\textbf{u}^{00}\textbf{u}_{0}-8(\textbf{u}^{00})^{2}+\hat{R}^{\textbf{q}}\\
&0\\
&\-g^{00}\textbf{u}_{0}^{lm}
\end{array}
\right).
\end{align*}

%---------------------------------------------------------------------------------------------------------------------

\subsection{The conformal fluid evolution}\label{section:3.3}

In this section, we transform the conformal fluid equation into a symmetric hyperbolic system. At first, we choose an appropriate conformal factor for the potential function $\Psi(\tau,x)$.

Define
\begin{equation}\label{4.29}
\widetilde{\del}^{\mu}\Psi=e^{-\lambda\Phi}\del^{\mu}\Psi=\tau^{\lambda}\del^{\mu}\Psi,
\end{equation}
where $\lambda$ will be determined later. Then
\begin{equation}\label{4.30}
\widetilde{h}=-\widetilde{g}_{\mu\nu}\widetilde{\del}^{\mu}\Psi\widetilde{\del}^{\nu}\Psi=-\tau^{2(\lambda-1)}g_{\mu\nu}\del^{\mu}\Psi\del^{\nu}\Psi
:=\tau^{2(\lambda-1)}h.
\end{equation}
In terms of $h$ and $\del\Psi$, from \eqref{4.29}, \eqref{4.30}, we have by neglecting some unnecessary constants $A^{\frac{2}{\alpha+1}}I^{-\frac{\alpha}{\alpha+1}}$
\begin{equation}\label{4.31}
\widetilde{T}^{\mu\nu}=\frac{[\tau^{2(\lambda-1)}h]^{\frac{1-\alpha}{2\alpha}}\tau^{2\lambda}\del^{\mu}\Psi\del^{\nu}\Psi}{[I-(\tau^{2(\lambda-1)}h)
^{\frac{\alpha+1}{2\alpha}}]^{\frac{1}{\alpha+1}}}
-[I-(\tau^{2(\lambda-1)}h)^{\frac{\alpha+1}{2\alpha}}]^{\frac{\alpha}{\alpha+1}}\tau^{2}g^{\mu\nu}
\end{equation}
and so
\begin{equation}\label{4.32}
\widetilde{T}=3p-\rho=-\frac{4I-3[\tau^{2(\lambda-1)}h]^{\frac{\alpha+1}{2\alpha}}}{[I-(\tau^{2(\lambda-1)}h)^{\frac{\alpha+1}{2\alpha}}]^{\frac{1}{\alpha+1}}}.
\end{equation}
Then we have by \eqref{1.11} and \eqref{4.31}-\eqref{4.32}
\begin{eqnarray}\label{4.33}
\nabla_{\mu}\widetilde{T}^{\mu\nu}&=&\nabla_{\mu}\left(\frac{[\tau^{2(\lambda-1)}h]^{\frac{1-\alpha}{2\alpha}}\tau^{2\lambda}\del^{\mu}\Psi\del^{\nu}\Psi}{[I-(\tau^{2(\lambda-1)}h)
^{\frac{\alpha+1}{2\alpha}}]^{\frac{1}{\alpha+1}}}
-[I-(\tau^{2(\lambda-1)}h)^{\frac{\alpha+1}{2\alpha}}]^{\frac{\alpha}{\alpha+1}}\tau^{2}g^{\mu\nu}\right)\nonumber\\
&=&\nabla_{\mu}\left(\frac{(\tau^{2(\lambda-1)}h)^{\frac{1-\alpha}{2\alpha}}\tau^{2\lambda}\del^{\mu}\Psi}
{[I-(\tau^{2(\lambda-1)}h)^{\frac{\alpha+1}{2\alpha}}]^{\frac{1}{\alpha+1}}}\right)\del^{\nu}\Psi
+\frac{(\lambda-1)(\tau^{2(\lambda-1)}h)^{\frac{1-\alpha}{2\alpha}}\tau^{2\lambda-1}hg^{0\nu}}
{[I-(\tau^{2(\lambda-1)}h)^{\frac{\alpha+1}{2\alpha}}]^{\frac{1}{\alpha+1}}}\nonumber\\
&&-[I-(\tau^{2(\lambda-1)}h)^{\frac{\alpha+1}{2\alpha}}]^{\frac{\alpha}{\alpha+1}}2\tau g^{0\nu}\nonumber\\
&=&\frac{6}{\tau}\left(\frac{[\tau^{2(\lambda-1)}h]^{\frac{1-\alpha}{2\alpha}}\tau^{2\lambda}\del^{\tau}\Psi\del^{\nu}\Psi}{[I-(\tau^{2(\lambda-1)}h)
^{\frac{\alpha+1}{2\alpha}}]^{\frac{1}{\alpha+1}}}
-[I-(\tau^{2(\lambda-1)}h)^{\frac{\alpha+1}{2\alpha}}]^{\frac{\alpha}{\alpha+1}}\tau^{2}g^{0\nu}\right)\nonumber\\
&&+\tau g^{0\nu}\frac{4I-3[\tau^{2(\lambda-1)}h]^{\frac{\alpha+1}{2\alpha}}}{[I-(\tau^{2(\lambda-1)}h)^{\frac{\alpha+1}{2\alpha}}]^{\frac{1}{\alpha+1}}}.
\end{eqnarray}
From \eqref{4.33}, we get easily
\begin{eqnarray}\label{4.34}
&\;&\nabla_{\mu}\left(\frac{(\tau^{2(\lambda-1)}h)^{\frac{1-\alpha}{2\alpha}}\tau^{2\lambda}\del^{\mu}\Psi}
{[I-(\tau^{2(\lambda-1)}h)^{\frac{\alpha+1}{2\alpha}}]^{\frac{1}{\alpha+1}}}\right)\del^{\nu}\Psi\nonumber\\
&=&\frac{6}{\tau}\left(\frac{[\tau^{2(\lambda-1)}h]^{\frac{1-\alpha}{2\alpha}}\tau^{2\lambda}\del^{\tau}\Psi}{[I-(\tau^{2(\lambda-1)}h)
^{\frac{\alpha+1}{2\alpha}}]^{\frac{1}{\alpha+1}}}\right)\del^{\nu}\Psi+\tau g^{0\nu}\frac{(2-\lambda)(\tau^{2(\lambda-1)}h)^{\frac{\alpha+1}{2\alpha}}}{[I-(\tau^{2(\lambda-1)}h)^{\frac{\alpha+1}{2\alpha}}]^{\frac{1}{\alpha+1}}}.
\end{eqnarray}
Contracting \eqref{4.34} with $\del_{\nu}\Psi$ and using the fact $h=-\del^{\nu}\Psi\del_{\nu}\Psi\geq0$, we have
\begin{eqnarray}\label{4.35}
&&\nabla_{\mu}\left(\frac{(\tau^{2(\lambda-1)}h)^{\frac{1-\alpha}{2\alpha}}\tau^{2\lambda}\del^{\mu}\Psi}
{[I-(\tau^{2(\lambda-1)}h)^{\frac{\alpha+1}{2\alpha}}]^{\frac{1}{\alpha+1}}}\right)\nonumber\\
&=&\frac{6}{\tau}\left(\frac{[\tau^{2(\lambda-1)}h]^{\frac{1-\alpha}{2\alpha}}\tau^{2\lambda}\del^{\tau}\Psi}{[I-(\tau^{2(\lambda-1)}h)
^{\frac{\alpha+1}{2\alpha}}]^{\frac{1}{\alpha+1}}}\right)
-\frac{(2-\lambda)}{\tau}\frac{(\tau^{2(\lambda-1)}h)^{\frac{1-\alpha}{2\alpha}}\tau^{2\lambda}\del^{\tau}\Psi}
{[I-(\tau^{2(\lambda-1)}h)^{\frac{\alpha+1}{2\alpha}}]^{\frac{1}{\alpha+1}}}\nonumber\\
&=&\frac{4+\lambda}{\tau}\frac{(\tau^{2(\lambda-1)}h)^{\frac{1-\alpha}{2\alpha}}\tau^{2\lambda}\del^{\tau}\Psi}
{[I-(\tau^{2(\lambda-1)}h)^{\frac{\alpha+1}{2\alpha}}]^{\frac{1}{\alpha+1}}}
\end{eqnarray}
Expanding \eqref{4.35} directly gives
\begin{eqnarray}\label{4.36}
&&[I-(\tau^{2(\lambda-1)}h)^{\frac{1+\alpha}{2\alpha}}]\left(\Box_{g}\Psi-\frac{1-\alpha}{\alpha}\frac{\del^{\mu}\Psi\del^{\nu}\Psi}{h}
\nabla_{\mu}\nabla_{\nu}\Psi\right)\nonumber\\
&&-\frac{1}{\alpha}(\tau^{2(\lambda-1)}h)^{\frac{1-\alpha}{2\alpha}}\tau^{2(\lambda-1)}\del^{\mu}\Psi\del^{\nu}\Psi\nabla_{\mu}\nabla_{\nu}\Psi
+\frac{1}{\alpha}(\lambda-1)(\tau^{2(\lambda-1)}h)^{\frac{1-\alpha}{2\alpha}}\tau^{2\lambda-3}h\nonumber\\
&=&\frac{4-\lambda-2(\lambda-1)\frac{1-\alpha}{2\alpha}}{\tau}\del^{\tau}\Psi[I-(\tau^{2(\lambda-1)}h)^{\frac{1+\alpha}{2\alpha}}].
\end{eqnarray}
We choose $\lambda$ such that
$$4-\lambda-2(\lambda-1)\frac{1-\alpha}{2\alpha}=0,$$
and thus
 $\lambda=3\alpha+1.$

\begin{remark}\label{remark:3.3}
If $\lambda=3\alpha+1$, then $\widetilde{\del}\Psi$ has the same asymptotic behavior as the background solution $\overline{\Psi}(\tau)$ when $\tau\rightarrow0$. This property plays very important roles in the non-degeneracy of above wave equation.
\end{remark}
Define a function $\Theta(\tau,x)$ such that
\begin{equation}\label{4.37}
\widetilde{\del}^{\mu}\Psi-\widetilde{\del}^{\mu}\overline{\Psi}=\tau^{3\alpha+1}\del^{\mu}\Theta.
\end{equation}
Then by \eqref{4.29}
\begin{equation}\label{4.38}
\widetilde{\del}^{\tau}\Psi=\tau^{3\alpha+1}\del^{\tau}\Psi
=\tau^{3\alpha+1}\left(\del^{\tau}\Theta+\frac{(IK)^{\frac{\alpha}{\alpha+1}}w}{(1+K\tau^{3(\alpha+1)})^{\frac{\alpha}{\alpha+1}}}\right)
:=\tau^{3\alpha+1}(\del^{\tau}\Theta+f(\tau)).
\end{equation}
We have
\begin{equation}\label{4.39}
\widetilde{\del}^{i}\Psi=\tau^{3\alpha+1}\del^{i}\Theta.
\end{equation}
Based on \eqref{4.37}-\eqref{4.39}, we have
\begin{eqnarray}\label{4.40}
h&=&-g_{00}(\del^{\tau}\Psi)^{2}-2g_{0i}\del^{\tau}\Psi\del^{i}\Psi -g_{ij}\del^{i}\Psi\del^{j}\Psi\nonumber\\
&=&-g_{00}f^{2}(\tau)-2g_{0\mu}f(\tau)\del^{\mu}\Theta-g_{00}(\del^{\tau}\Theta)^{2}-2g_{0i}\del^{i}\Theta\del^{\tau}\Theta
-g_{ij}\del^{i}\Theta\del^{j}\Theta.
\end{eqnarray}
\begin{remark}\label{remark:3.4}
It is obvious that $(\widetilde{\eta},\widetilde{\del}^{\tau}\overline{\Psi})$ defined in Section \ref{section:2.3} is a solution to \eqref{4.36}, which means that $(\eta,\tau^{3\alpha+1}f(\tau))$ is the solution to \eqref{4.36}.
\end{remark}
Define
\begin{eqnarray*}
B^{00}&=&(I-(\tau^{2(\lambda-1)}h)^{\frac{1+\alpha}{2\alpha}})(g^{00}-\frac{1-\alpha}{\alpha}\frac{\del^{\tau}\Psi\del^{\tau}\Psi}{h})
-\frac{1}{\alpha}(\tau^{2(\lambda-1)}h)^{\frac{1-\alpha}{2\alpha}}\tau^{2(\lambda-1)}\del^{\tau}\Psi\del^{\tau}\Psi,\\
B^{0i}&=&(I-(\tau^{2(\lambda-1)}h)^{\frac{1+\alpha}{2\alpha}})(g^{0i}-\frac{1-\alpha}{\alpha}\frac{\del^{\tau}\Psi\del^{i}\Psi}{h})
-\frac{1}{\alpha}(\tau^{2(\lambda-1)}h)^{\frac{1-\alpha}{2\alpha}}\tau^{2(\lambda-1)}\del^{\tau}\Psi\del^{i}\Psi,\\
B^{jk}&=&(I-(\tau^{2(\lambda-1)}h)^{\frac{1+\alpha}{2\alpha}})(g^{jk}-\frac{1-\alpha}{\alpha}\frac{\del^{j}\Psi\del^{k}\Psi}{h})
-\frac{1}{\alpha}(\tau^{2(\lambda-1)}h)^{\frac{1-\alpha}{2\alpha}}\tau^{2(\lambda-1)}\del^{j}\Psi\del^{k}\Psi
\end{eqnarray*}
and similarly
\begin{eqnarray*}
\overline{B}^{00}&=&(I-(\tau^{2(\lambda-1)}\overline{h})^{\frac{1+\alpha}{2\alpha}})(\eta^{00}-\frac{1-\alpha}{\alpha}\frac{\del^{\tau}\overline{\Psi}\del^{\tau}\overline{\Psi}}{\overline{h}})
-\frac{1}{\alpha}(\tau^{2(\lambda-1)}\overline{h})^{\frac{1-\alpha}{2\alpha}}\tau^{2(\lambda-1)}\del^{\tau}\overline{\Psi}\del^{\tau}\overline{\Psi},\\
\overline{B}^{0i}&=&(I-(\tau^{2(\lambda-1)}\overline{h})^{\frac{1+\alpha}{2\alpha}})(\eta^{0i}-\frac{1-\alpha}{\alpha}\frac{\del^{\tau}\overline{\Psi}\del^{i}\overline{\Psi}}{\overline{h}})
-\frac{1}{\alpha}(\tau^{2(\lambda-1)}\overline{h})^{\frac{1-\alpha}{2\alpha}}\tau^{2(\lambda-1)}\del^{\tau}\overline{\Psi}\del^{i}\overline{\Psi},\\
\overline{B}^{jk}&=&(I-(\tau^{2(\lambda-1)}\overline{h})^{\frac{1+\alpha}{2\alpha}})(\eta^{jk}-\frac{1-\alpha}{\alpha}\frac{\del^{j}\overline{\Psi}\del^{k}\overline{\Psi}}{\overline{h}})
-\frac{1}{\alpha}(\tau^{2(\lambda-1)}\overline{h})^{\frac{1-\alpha}{2\alpha}}\tau^{2(\lambda-1)}\del^{j}\overline{\Psi}\del^{k}\overline{\Psi},
\end{eqnarray*}
where $\overline{h}=-\eta^{00}\del_{\tau}\overline{\Psi}\del_{\tau}\overline{\Psi}$.

Define
\begin{equation}\label{4.41}
\textbf{p}_{\mu}=\del_{\mu}\Theta=g_{0\nu}\del^{\nu}\Theta.
\end{equation}
Expanding \eqref{4.36} in terms of $\Theta$, we have
\begin{eqnarray}\label{4.42}
&&B^{00}\del_{\tau}(\del_{\tau}\Theta)+B^{00}\del_{\tau}(g_{00}f(\tau))+2B^{0i}\del_{i}\del_{\tau}\Theta
+2B^{0i}\del_{i}(g_{00}f(\tau))\nonumber\\
&&+B^{ij}\del_{i}\del_{j}\Theta+B^{ij}\del_{i}(g_{j0}f(\tau))+T=0
\end{eqnarray}
and 
\begin{eqnarray*}
T&=&\left[-\Gamma^{\kappa}(\del_{\kappa}\Theta+g_{\kappa0}f(\tau))+\frac{1-\alpha}{\alpha}
\frac{\del^{\mu}\Psi\del^{\nu}\Psi}{h}\Gamma^{\kappa}_{\mu\nu}(\del_{\kappa}\Theta+g_{\kappa0}f(\tau))\right]
[I-(\tau^{2(\lambda-1)}h)^{\frac{1+\alpha}{2\alpha}}]\nonumber\\
&&+\frac{1}{\alpha}(\tau^{2(\lambda-1)}h)^{\frac{1-\alpha}{2\alpha}}[\tau^{2(\lambda-1)}\del^{\mu}\Psi\del^{\nu}\Psi\Gamma^{\kappa}_{\mu\nu}(\del_{\kappa}\Theta+g_{\kappa0}f(\tau))
+(\lambda-1)\tau^{2\lambda-3}h].
\end{eqnarray*}
From \eqref{4.42}, we get
\begin{equation}\label{4.43}
B^{00}\del_{\tau}\textbf{p}_{0}+2B^{0i}\del_{i}\textbf{p}_{0}+B^{ij}\del_{i}\textbf{p}_{j}=\hat{T}-\hat{\overline{T}},
\end{equation}
where
\begin{eqnarray}\label{4.44}
\hat{T}&=&-T-B^{00}\del_{\tau}(g_{00}f(\tau))-2B^{0i}\del_{i}(g_{00}f(\tau))-B^{ij}\del_{i}(g_{j0}f(\tau))\nonumber\\
&=&\left[\Gamma^{\kappa}(\del_{\kappa}\Theta+g_{\kappa0}f(\tau))-\frac{1-\alpha}{\alpha}
\frac{\del^{\mu}\Psi\del^{\nu}\Psi}{h}\Gamma^{\kappa}_{\mu\nu}(\del_{\kappa}\Theta+g_{\kappa0}f(\tau))\right]
[I-(\tau^{2(\lambda-1)}h)^{\frac{1+\alpha}{2\alpha}}]\nonumber\\
&&-\frac{1}{\alpha}(\tau^{2(\lambda-1)}h)^{\frac{1-\alpha}{2\alpha}}[\tau^{2(\lambda-1)}\del^{\mu}\Psi\del^{\nu}\Psi\Gamma^{\kappa}_{\mu\nu}(\del_{\kappa}\Theta+g_{\kappa0}f(\tau))
+(\lambda-1)\tau^{2\lambda-3}h]\nonumber\\
&&-B^{00}\del_{\tau}(g_{00}f(\tau))-2B^{0i}\del_{i}(g_{00}f(\tau))-B^{ij}\del_{i}(g_{j0}f(\tau)).
\end{eqnarray}
When we use the background metric $\eta$ instead of $g$ and $\overline{\Psi}$ instead of $\Psi$ to calculate $\hat{T}$, we get $\hat{\overline{T}}$. Clearly $\hat{\overline{T}}=0$ since $(\eta,\overline{\Psi})$ is the solution to the fluid equation. From \eqref{4.36}, \eqref{4.43} and \eqref{4.44}, we can easily rewrite the fluid equation into the following symmetric hyperbolic system
\begin{eqnarray}\label{4.45}
B^{\kappa}\del_{\kappa}\left(
\begin{array}{ll}
\textbf{p}_{0}\\
\textbf{p}_{1}\\
\textbf{p}_{2}\\
\textbf{p}_{3}
\end{array}\right)=
\left(\begin{array}{cc}
&\hat{T}-\hat{\overline{T}}\\&0\\&0\\&0
\end{array}
\right),
\end{eqnarray}
where
\begin{align*}
B^{0}=\left(
\begin{array}{cccc}
B^{00}&0&0&0\\
0&-B^{11}&-B^{21}&-B^{31}\\
0&-B^{12}&-B^{22}&-B^{32}\\
0&-B^{13}&-B^{23}&-B^{33}
\end{array}
\right),\quad
B^{i}=\left(
\begin{array}{cccc}\label{4.61}
2B^{0i}&B^{i1}&B^{i2}&B^{i3}\\
B^{i1}&0&0&0\\
B^{i2}&0&0&0\\
B^{i3}&0&0&0
\end{array}
\right).
\end{align*}

%----------------------------------------------------------------------------------------------------------------------

\subsection{A key argument for the source terms}\label{section:3.4}

In order to analyze the structure of the symmetric hyperbolic system, in this section, we mainly focus on the source terms of above two subsections, especially the terms
$\hat{M}^{0\mu}$, $\hat{M}^{ij}$, $\acute{M}^{ij}$, $\hat{R}^{\textbf{q}}$ and $\hat{T}-\hat{\overline{T}}$.

We need the following basic lemmas. At first we have the following algebraic relationship between $g^{-1}$ and $g$.
\begin{lemma}\label{lemma:3.5}
Assume $g^{-1}=(g^{\mu\nu})$ is a symmetric $(1+3)\times(1+3)$ Lorentz metric with $g^{00}<0$ and $(g^{ij})$ positive definite, then
\begin{eqnarray}\label{4.46}
g_{00}&=&\frac{1}{g^{00}-d^{2}},\\\label{4.47}
g_{0i}&=&\frac{g_{ij}g^{0j}}{d^{2}-g^{00}},
\end{eqnarray}
where
$d^{2}=g_{ij}g^{0i}g^{0j}$.
\end{lemma}
\begin{proof}
The proof can be found in Lemmas 1 and 2 of \cite{R}.
\end{proof}
The following two lemmas will be repeatedly used in this section.
\begin{lemma}\label{lemma:3.6}
Suppose that $a_{i}\;(i=1,\cdots,n)$ and $\overline{a}_{i}\,(i=1,\cdots,n)$ are smooth functions, then we have
\begin{equation}\label{4.48}
\prod_{i=1}^{n}a_{i}-\prod_{i=1}^{n}\overline{a}_{i}=\sum_{j=1}^{n}F^{j}(\overline{a}_{i},a_{i}-\overline{a}_{i}),
\end{equation}
here
$$
F^{j}(a_{i},a_{i}-\overline{a}_{i})=\prod_{k=1}^{j}\prod_{l=j+1}^{n}(a_{i_{k}}-\overline{a}_{i_{k}})\overline{a}_{i_{l}},
$$
where $i_{k}\in({1,\cdots,n })$.
\end{lemma}
\begin{proof}
We can get this result by induction. At first, for $i=1$, it holds obviously. Assume that \eqref{4.48} holds for $i=n-1$, namely
\begin{equation*}
\prod_{i=1}^{n-1}a_{i}-\prod_{i=1}^{n-1}\overline{a}_{i}=\sum_{j=1}^{n-1}F^{j}(\overline{a}_{i},a_{i}-\overline{a}_{i}),
\end{equation*}
then for $i=n$, we have
\begin{eqnarray*}
\prod_{i=1}^{n}a_{i}-\prod_{i=1}^{n}\overline{a}_{i}&=&a_{n}\prod_{i=1}^{n-1}a_{i}-\overline{a}_{n}\prod_{i=1}^{n-1}\overline{a}_{i}\\
&=&(a_{n}-\overline{a}_{n})(\prod_{i=1}^{n-1}a_{i}-\prod_{i=1}^{n-1}\overline{a}_{i})+\overline{a}_{n}(\prod_{i=1}^{n-1}a_{i}-\prod_{i=1}^{n-1}\overline{a}_{i})
+
(a^{n}-\overline{a}_{n})\prod_{i=1}^{n-1}\overline{a}_{i}\\
&=&(a_{n}-\overline{a}_{n})(\sum_{j=1}^{n-1}F^{j}(\overline{a}_{i},a_{i}-\overline{a}_{i}))+\overline{a}_{n}(\sum_{j=1}^{n-1}F^{j}(\overline{a}_{i},a_{i}-\overline{a}_{i}))
+(a^{n}-\overline{a}_{n})\prod_{i=1}^{n-1}\overline{a}_{i}\\
&=&\sum_{j=1}^{n}F^{j}(\overline{a}_{i},a_{i}-\overline{a}_{i}).
\end{eqnarray*}
\end{proof}

The proof of the following result is straighforward

\begin{lemma}\label{lemma:3.7}
Let $f(x)$ be analytical in the neighborhood of a point $\overline{x}$, and assume that $f^{'}(\overline{x})\neq0$, then there exists a small parameter $\delta$, such that when $x\in[\overline{x}-\delta,\overline{x}+\delta]$, we have
\begin{equation}\label{4.49}
f(x)-f(\overline{x})\sim f^{'}(\overline{x})(x-\overline{x}).
\end{equation}
\end{lemma}

Based on Lemmas \ref{lemma:3.5}-\ref{lemma:3.7}, we have the following important estimates in terms of the unknowns.

\begin{lemma}\label{lemma:3.8}
In terms of the unknowns $(\textbf{u}^{0\mu},\textbf{u}^{0\mu}_{\nu},\textbf{u}^{ij},\textbf{u}^{ij}_{\mu},\textbf{q},\textbf{q}_{\mu},\textbf{p}_{\mu})$, we have
\begin{eqnarray*}
g^{ij}-\eta^{ij}&\sim&\textbf{u}^{ij}+\frac{\eta^{ij}}{w^{2}}(2\tau \textbf{u}^{00}-\textbf{q}),\\
\del_{\lambda}(g^{ij}-\eta^{ij})&\sim& \textbf{u}_{\lambda}^{ij}
+\frac{2\dot{w}\eta^{ij}\delta_{\lambda}^{0}}{\tau w^{4}}(2\tau\textbf{u}^{00}-\textbf{q})\\
&&+\frac{\eta^{ij}}{w^{2}}(2\textbf{u}^{00}\delta_{\lambda}^{0}+2\tau \textbf{u}^{00}_{i}\delta_{\lambda}^{i}+(\textbf{u}_{0}^{00}+\textbf{u}^{00})\delta^{0}_{\lambda}-\textbf{q}_{\lambda}),\\
g_{ij}-\eta_{ij}&\sim&\sum_{k=1}^{3}F^{k}(\eta^{ij},g^{ij}-\eta^{ij})+(\eta^{ij})^{*}\frac{3(2\tau \textbf{u}^{00}-\textbf{q})}{w^{2}},\\
g_{00}-\eta_{00}&\sim&\frac{-2\tau \textbf{u}^{00}+\sum_{k=1}^{3}F^{k}(\eta,\eta^{-1},(g_{ij}-\eta_{ij}),(g^{0i}-\eta^{0i}))}{(\eta^{00})^{2}},\\
g_{0i}-\eta_{0i}&\sim& \frac{\sum_{k=1}^{3}F^{k}(\eta,\eta^{-1},(g_{ij}-\eta_{ij}),(g^{0i}-\eta^{0i}),(g^{00}-\eta^{00}))}{(\eta^{00})^{2}},\\
\del^{\tau}\Psi-\del^{\tau}\overline{\Psi}&=&2\tau \textbf{u}^{0\nu}\textbf{p}_{\nu}+\eta^{0\nu}\textbf{p}_{\nu},\\
\del^{i}\Psi-\del^{i}\overline{\Psi}&\sim& (g^{i\nu}-\eta^{i\nu})\textbf{p}_{\nu}+\eta^{i\nu}\textbf{p}_{\nu},\\
h-\overline{h}&=&\sum_{k=1}^{3}F^{k}(\eta,\del^{\mu}\overline{\Psi},(g_{\mu\nu}-\eta_{\mu\nu}),(\del^{\mu}\Psi-\del^{\mu}\overline{\Psi})),
\end{eqnarray*}
provided that $\|(\textbf{u}^{0\mu},\textbf{u}^{0\mu}_{\nu},\textbf{u}^{ij},\textbf{u}^{ij}_{\mu},\textbf{q},\textbf{q}_{\mu},\textbf{p}_{\mu})\|_{L^{\infty}(\mathbb T^{3})}$ is sufficiently small. In above, $(\eta^{ij})^{*}$ denotes the cofactor of $\eta^{ij}$ of the positive definite matrix $(\eta^{ij})$.
\end{lemma}
\begin{proof}
At first, from the definition of $\textbf{g}^{ij}$ and \textbf{q}, i.e. \eqref{4.9}-\eqref{4.10}, we have
\begin{equation}\label{4.50}
det({g^{ij}})=e^{\frac{3(2\tau \textbf{u}^{00}-\textbf{q})}{w^{2}}},
\end{equation}
and
\begin{equation}\label{4.51}
g^{ij}=e^{\frac{2\tau \textbf{u}^{00}-\textbf{q}}{w^{2}}}\textbf{g}^{ij}.
\end{equation}
Then we have by \eqref{4.51}
\begin{eqnarray*}
g^{ij}-\eta^{ij}&=&e^{\frac{2\tau \textbf{u}^{00}-\textbf{q}}{w^{2}}}\textbf{g}^{ij}-\eta^{ij}\\
&=&e^{\frac{2\tau \textbf{u}^{00}-\textbf{q}}{w^{2}}}(\textbf{g}^{ij}-\eta^{ij})+\eta^{ij}(e^{\frac{2\tau \textbf{u}^{00}-\textbf{q}}{w^{2}}}-1)\\
&\sim&\textbf{u}^{ij}+\eta^{ij}\frac{2\tau \textbf{u}^{00}-\textbf{q}}{w^{2}}.
\end{eqnarray*}
And we have
\begin{eqnarray*}
\del_{\lambda}(g^{ij}-\eta^{ij})&=&\del_{\lambda}(e^{\frac{2\tau \textbf{u}^{00}-\textbf{q}}{w^{2}}}\textbf{g}^{ij}-\eta^{ij})\\
&\sim& \textbf{u}_{\lambda}^{ij}+\frac{2\dot{w}\eta^{ij}\delta_{\lambda}^{0}}{\tau w^{4}}(2\tau\textbf{u}^{00}-\textbf{q})+\frac{\eta^{ij}}{w^{2}}(2\textbf{u}^{00}\delta_{\lambda}^{0}+2\tau\del_{\lambda}\textbf{u}^{00}-\textbf{q}_{\lambda}).
\end{eqnarray*}
Since
$
\del_{\tau}\textbf{u}^{00}=\frac{\textbf{u}_{0}^{00}+\textbf{u}^{00}}{2\tau}
$
and
$
\del_{i}\textbf{u}^{00}=\textbf{u}^{00}_{i}
$
we obtain
\begin{eqnarray*}
\del_{\lambda}(g^{ij}-\eta^{ij})
&\sim& \textbf{u}_{\lambda}^{ij}
+\frac{2\dot{w}\eta^{ij}\delta_{\lambda}^{0}}{\tau w^{4}}(2\tau\textbf{u}^{00}-\textbf{q})\\
&&+\frac{\eta^{ij}}{w^{2}}(2\textbf{u}^{00}\delta_{\lambda}^{0}
+2\tau \textbf{u}^{00}_{i}\delta_{\lambda}^{i}+(\textbf{u}_{0}^{00}+\textbf{u}^{00})\delta^{0}_{\lambda}-\textbf{q}_{\lambda}).
\end{eqnarray*}
According to \eqref{4.50} and \eqref{4.51}, we have
\begin{equation}\label{4.52}
g_{ij}=\frac{(g^{ij})^{*}}{det(g^{ij})}=\frac{(g^{ij})^{*}}{e^{\frac{3(2\tau \textbf{u}^{00}-\textbf{q})}{w^{2}}}}.
\end{equation}
From \eqref{4.52}, we easily see
\begin{eqnarray*}
g_{ij}-\eta_{ij}&=&\frac{(g^{ij})^{*}}{e^{\frac{3(2\tau \textbf{u}^{00}-\textbf{q})}{w^{2}}}}-(\eta^{ij})^{*}\\
&=&\frac{(g^{ij})^{*}-(\eta^{ij})^{*}+(\eta^{ij})^{*}(1-e^{\frac{3(2\tau \textbf{u}^{00}-\textbf{q})}{w^{2}}})}{e^{\frac{3(2\tau \textbf{u}^{00}-\textbf{q})}{w^{2}}}}\\
&\sim& \sum_{k=1}^{2}F^{k}(\eta,(g^{ij}-\eta^{ij}))+(\eta^{ij})^{*}\left(\frac{3(2\tau \textbf{u}^{00}-q)}{w^{2}}\right).
\end{eqnarray*}
In the second equality, we have used Lemma \ref{lemma:3.5}, since $(g^{ij})^{*}=(g^{-1})\times(g^{-1})$.

Before we estimating $g_{0\mu}-\eta_{0\mu}$, we have to estimate $d^{2}$, obviously, $\overline{d}^{2}=0$, thus
\begin{equation}\label{4.53}
d^{2}-\overline{d}^{2}=g_{ij}g^{0i}g^{0j}-\eta_{ij}\eta^{0i}\eta^{0j}=\sum_{k=1}^{3}F^{k}(\eta,\eta^{-1},(g_{ij}-\eta_{ij}),(g^{0i}-\eta^{0i})).
\end{equation}
Utilizing \eqref{4.53} and Lemma \ref{lemma:3.5}, we have
\begin{eqnarray*}
g_{00}-\eta_{00}&=&\frac{1}{g^{00}-d^{2}}-\frac{1}{\eta^{00}-\overline{d}^{2}}\\
&=&\frac{\eta^{00}-g^{00}-\overline{d}^{2}+d^{2}}{(g^{00}-d^{2})(\eta^{00}-\overline{d}^{2})}\\
&\sim&\frac{-2\tau \textbf{u}^{00}+\sum_{k=1}^{3}F^{k}(\eta,\eta^{-1},(g_{ij}-\eta_{ij}),(g^{0i}-\eta^{0i}))}{(g^{00}-\eta^{00}+\eta^{00}-d^{2})(\eta^{00}-\overline{d}^{2})}\\
&\sim&\frac{-2\tau \textbf{u}^{00}+\sum_{k=1}^{3}F^{k}(\eta,\eta^{-1},(g_{ij}-\eta_{ij}),(g^{0i}-\eta^{0i}))}{(\eta^{00})^{2}}
\end{eqnarray*}
and
\begin{eqnarray*}
g_{0i}-\eta_{0i}&=&\frac{g_{ij}g^{0j}}{d^{2}-g^{00}}-\frac{\eta_{ij}\eta^{0j}}{\overline{d}^{2}-\eta^{00}}\\
&=&\frac{g^{00}\eta_{ij}\eta^{0j}-\eta^{00}g_{ij}g^{0j}+g_{ij}g^{0i}\overline{d}^{2}-\eta_{ij}\eta^{0i}d^{2}}{(d^{2}-g^{00})(\overline{d}^{2}-\eta^{00})}\\
&\sim&\frac{\sum_{k=1}^{3}F^{k}(\eta,\eta^{-1},(g^{00}-\eta^{00}),(g_{ij}-\eta_{ij}),(g^{0j}-\eta^{0j}))}{(\eta^{00})^{2}}.
\end{eqnarray*}
For the fluid variables, we have the following
\begin{equation*}
\del^{\tau}\Psi-\del^{\tau}\overline{\Psi}=\del^{\tau}\Theta=(g^{0\nu}-\eta^{0\nu})\del_{\nu}\Theta+\eta^{0\nu}\del_{\nu}\Theta
=2\tau \textbf{u}^{00}\textbf{p}_{\nu}+\eta^{0\nu}\textbf{p}_{\nu}.
\end{equation*}
\begin{equation*}
\del^{i}\Psi-\del^{i}\overline{\Psi}=(g^{i\nu}-\eta^{i\nu})\del_{\nu}\Theta+\eta^{i\nu}\del_{\nu}\Theta
=(g^{i\nu}-\eta^{i\nu})\textbf{p}_{\nu}+\eta^{i\nu}\textbf{p}_{\nu},
\end{equation*}
and
\begin{eqnarray*}
h-\overline{h}=-g_{\mu\nu}\del^{\mu}\Psi\del^{\nu}\Psi+\eta_{\mu\nu}\del^{\mu}\overline{\Psi}\del^{\nu}\overline{\Psi}
=\sum_{k=1}^{3}F^{k}(\eta,\del^{\mu}\overline{\Psi},(g_{\mu\nu}-\eta_{\mu\nu}),(\del^{\mu}\Psi-\del^{\mu}\overline{\Psi})).
\end{eqnarray*}
\end{proof}
\begin{remark}\label{remark:3.9}
Above lemma shows that when $\|(\textbf{u}^{0\mu},\textbf{u}^{0\mu}_{\nu},\textbf{u}^{ij},\textbf{u}^{ij}_{\mu},\textbf{q},\textbf{q}_{\mu},\textbf{p}_{\mu})\|_{L^{\infty}(\mathbb T^{3})}$ is sufficiently small, the differences between the unknowns and the background solutions are small and depend only on the unknowns and $\tau$. Thus, for convenience of analysis, we denote the differences by
$$H(\tau,\textbf{u}):=\{f(\tau,\textbf{u})|f(\tau,0)=0\}$$
with $\textbf{u}=(\textbf{u}^{0\mu},\textbf{u}^{0\mu}_{\nu},\textbf{u}^{ij},\textbf{u}^{ij}_{\mu},\textbf{q},\textbf{q}_{\mu},\textbf{p}_{\mu})$ and $f(\tau,\textbf{u})$ denote smooth functions, which are regular with $\tau\in[0,1]$.
\end{remark}
Now, we turn to estimate the source terms, from the representation of the source terms of the Einstein equations, the most difficult part is to estimate the following two terms:
\begin{equation}\label{4.54}
\frac{I-\frac{1}{2}(\widetilde{h})^{\frac{\alpha+1}{2\alpha}}}{(I-(\widetilde{h})^{\frac{\alpha+1}{2\alpha}})^{\frac{1}{\alpha+1}}}
-\frac{I-\frac{1}{2}(\overline{\widetilde{h}})^{\frac{\alpha+1}{2\alpha}}}{(I-(\overline{\widetilde{h}})^{\frac{\alpha+1}{2\alpha}})^{\frac{1}{\alpha+1}}}
\end{equation}
and
\begin{equation}\label{4.55}
\frac{(\widetilde{h})^{\frac{1-\alpha}{2\alpha}}\widetilde{\del}^{\mu}\Psi\widetilde{\del}^{\mu}\Psi}{(I-(\widetilde{h})^{\frac{\alpha+1}{2\alpha}})^{\frac{1}{\alpha+1}}}
-\frac{(\overline{\widetilde{h}})^{\frac{1-\alpha}{2\alpha}}\widetilde{\del}^{\mu}\overline{\Psi}\widetilde{\del}^{\mu}\overline{\Psi}}{(I-(\overline{\widetilde{h}})^{\frac{\alpha+1}{2\alpha}})^{\frac{1}{\alpha+1}}}.
\end{equation}
Based on Lemmas \ref{lemma:3.5}-\ref{lemma:3.8}, we estimate \eqref{4.54}-\eqref{4.55} as follows.

\begin{lemma}\label{lemma:3.10}
Assume that $\|\textbf{u}\|_{L^{\infty}(\mathbb T^{3})}$ is sufficiently small, we have
\begin{eqnarray*}
\frac{I-\frac{1}{2}(\widetilde{h})^{\frac{\alpha+1}{2\alpha}}}{(I-(\widetilde{h})^{\frac{\alpha+1}{2\alpha}})^{\frac{1}{\alpha+1}}}
-\frac{I-\frac{1}{2}(\overline{\widetilde{h}})^{\frac{\alpha+1}{2\alpha}}}{(I-(\overline{\widetilde{h}})^{\frac{\alpha+1}{2\alpha}})^{\frac{1}{\alpha+1}}}
\sim \tau^{3(\alpha+1)}H(\tau,\textbf{u})
\end{eqnarray*}
and
\begin{eqnarray*}
\frac{(\widetilde{h})^{\frac{1-\alpha}{2\alpha}}\widetilde{\del}^{\mu}\Psi\widetilde{\del}^{\mu}\Psi}{(I-(\widetilde{h})^{\frac{\alpha+1}{2\alpha}})^{\frac{1}{\alpha+1}}}
-\frac{(\overline{\widetilde{h}})^{\frac{1-\alpha}{2\alpha}}\widetilde{\del}^{\mu}\overline{\Psi}\widetilde{\del}^{\mu}\overline{\Psi}}{(I-(\overline{\widetilde{h}})^{\frac{\alpha+1}{2\alpha}})^{\frac{1}{\alpha+1}}}
\sim(\tau^{3\alpha+5}+\tau^{6\alpha+8})H(\tau,\textbf{u}).
\end{eqnarray*}
\end{lemma}

\begin{proof}
From \eqref{4.29}, \eqref{4.30} and $\lambda=3\alpha+1$, we see
that $
\widetilde{h}=\tau^{6\alpha}h,\quad \overline{\widetilde{h}}=\tau^{6\alpha}\overline{h}
$
and
$$
\widetilde{\del}^{\mu}\Psi=\tau^{3\alpha+1}\del^{\mu}\Psi,\quad \widetilde{\del}^{\mu}\overline{\Psi}=\tau^{3\alpha+1}\del^{\mu}\overline{\Psi}.
$$
Then direct calculations give
\begin{eqnarray*}
&&\frac{I-\frac{1}{2}(\widetilde{h})^{\frac{\alpha+1}{2\alpha}}}{(I-(\widetilde{h})^{\frac{\alpha+1}{2\alpha}})^{\frac{1}{\alpha+1}}}
-\frac{I-\frac{1}{2}(\overline{\widetilde{h}})^{\frac{\alpha+1}{2\alpha}}}{(I-(\overline{\widetilde{h}})^{\frac{\alpha+1}{2\alpha}})^{\frac{1}{\alpha+1}}}\\
&=&\frac{-\frac{1}{2}\tau^{3(\alpha+1)}[h^{\frac{\alpha+1}{2\alpha}}-\overline{h}^{\frac{\alpha+1}{2\alpha}}]}{(I-(\tau^{6\alpha}h)^{\frac{\alpha+1}{2\alpha}})^{\frac{1}{\alpha+1}}}
+(I-\tau^{3(\alpha+1)}\overline{h}^{\frac{\alpha+1}{2\alpha}})\\
&&\times\frac{(I-\tau^{3(\alpha+1)}\overline{h}^{\frac{\alpha+1}{2\alpha}})^{\frac{1}{\alpha+1}}
-[I-\tau^{3(\alpha+1)}\overline{h}^{\frac{\alpha+1}{2\alpha}}+\tau^{3(\alpha+1)}(\overline{h}^{\frac{\alpha+1}{2\alpha}}-h^{\frac{\alpha+1}{2\alpha}})]^{\frac{1}{\alpha+1}}}
{[(I-\tau^{3(\alpha+1)}h^{\frac{\alpha+1}{2\alpha}})(I-\tau^{3(\alpha+1)}\overline{h}^{\frac{\alpha+1}{2\alpha}})]^{\frac{1}{\alpha+1}}}\\
&\sim&\frac{-\frac{1}{2}\tau^{3(\alpha+1)}[h^{\frac{\alpha+1}{2\alpha}}-\overline{h}^{\frac{\alpha+1}{2\alpha}}]}{(I-(\tau^{6\alpha}h)^{\frac{\alpha+1}{2\alpha}})^{\frac{1}{\alpha+1}}}
+\frac{(I-\tau^{3(\alpha+1)}\overline{h}^{\frac{\alpha+1}{2\alpha}})\frac{\tau^{3(\alpha+1)}}{\alpha+1}(\overline{h}^{\frac{\alpha+1}{2\alpha}}-h^{\frac{\alpha+1}{2\alpha}})}
{[(I-\tau^{3(\alpha+1)}h^{\frac{\alpha+1}{2\alpha}})(I-\tau^{3(\alpha+1)}\overline{h}^{\frac{\alpha+1}{2\alpha}})]^{\frac{1}{\alpha+1}}}\\
&\sim&\tau^{3(\alpha+1)}H(\tau,\textbf{u}).
\end{eqnarray*}
In the fifth line, we have used Lemma \ref{lemma:3.7} with $f(x)=x^{\frac{1}{\alpha+1}}$ and in the last line we have used $f(x)=x^{\frac{\alpha+1}{2\alpha}}$ and Lemma \ref{lemma:3.8}.

For \eqref{4.55}, we have
\begin{eqnarray*}
&&\frac{(\widetilde{h})^{\frac{1-\alpha}{2\alpha}}\widetilde{\del}^{\mu}\Psi\widetilde{\del}^{\nu}\Psi}{(I-(\widetilde{h})^{\frac{\alpha+1}{2\alpha}})^{\frac{1}{\alpha+1}}}
-\frac{(\overline{\widetilde{h}})^{\frac{1-\alpha}{2\alpha}}\widetilde{\del}^{\mu}\overline{\Psi}\widetilde{\del}^{\nu}\overline{\Psi}}{(I-(\overline{\widetilde{h}})^{\frac{\alpha+1}{2\alpha}})^{\frac{1}{\alpha+1}}}\\
&=&\frac{\tau^{3\alpha+5}[h^{\frac{1-\alpha}{2\alpha}}\del^{\mu}\Psi\del^{\nu}\Psi-\overline{h}^{\frac{1-\alpha}{2\alpha}}\del^{\mu}\overline{\Psi}\del^{\nu}\overline{\Psi}]}
{(I-\tau^{3(\alpha+1)}h^{\frac{\alpha+1}{2\alpha}})^{\frac{1}{\alpha+1}}}\\
&&+\frac{\tau^{3\alpha+5}\overline{h}^{\frac{1-\alpha}{2\alpha}}\del^{\mu}\overline{\Psi}\del^{\nu}\overline{\Psi}
[(I-\tau^{3(\alpha+1)}\overline{h}^{\frac{\alpha+1}{2\alpha}})^{\frac{1}{\alpha+1}}
-(I-\tau^{3(\alpha+1)}h^{\frac{\alpha+1}{2\alpha}})^{\frac{1}{\alpha+1}}]}
{[(I-\tau^{3(\alpha+1)}\overline{h}^{\frac{\alpha+1}{2\alpha}})(I-\tau^{3(\alpha+1)}h^{\frac{\alpha+1}{2\alpha}})]^{\frac{1}{\alpha+1}}}\\
&\sim&\frac{\tau^{3\alpha+5}[h^{\frac{1-\alpha}{2\alpha}}\del^{\mu}\Psi\del^{\nu}\Psi-\overline{h}^{\frac{1-\alpha}{2\alpha}}\del^{\mu}\overline{\Psi}\del^{\nu}\overline{\Psi}]}
{(I-\tau^{3(\alpha+1)}h^{\frac{\alpha+1}{2\alpha}})^{\frac{1}{\alpha+1}}}\\
&&+\frac{\tau^{6\alpha+8}\overline{h}^{\frac{1-\alpha}{2\alpha}}\del^{\mu}\overline{\Psi}\del^{\nu}\overline{\Psi}
[\overline{h}^{\frac{1+\alpha}{2\alpha}}-h^{\frac{1+\alpha}{2\alpha}}]}
{(\alpha+1)[(I-\tau^{3(\alpha+1)}\overline{h}^{\frac{\alpha+1}{2\alpha}})(I-\tau^{3(\alpha+1)}h^{\frac{\alpha+1}{2\alpha}})]^{\frac{1}{\alpha+1}}}\\
&\sim&(\tau^{3\alpha+5}+\tau^{6\alpha+8})H(\tau,\textbf{u}).
\end{eqnarray*}
In the fifth line, we have used Lemma \ref{lemma:3.7} with $f(x)=x^{\frac{1}{\alpha+1}}$ and in the last line we have used Lemma \ref{lemma:3.6} and \ref{lemma:3.8}.
\end{proof}

Since $Q^{\mu\nu}(g,\del g)$ is quadratic in $(\del g^{\mu\nu})$ and analytical in $(g^{\mu\nu})$, by Lemmas \ref{lemma:3.6} and \ref{lemma:3.8}, we easily get the following statement.

\begin{lemma}\label{lemma:3.11}
We have \begin{equation*}
Q^{\mu\nu}(g,\del g)-Q^{\mu\nu}(\eta,\del\eta)\sim H(\tau,\textbf{u}).
\end{equation*}
\end{lemma}
The other terms in $\hat{M}^{\mu\nu}$, $\acute{M}^{ij}$ and $\hat{R}^{\textbf{q}}$ can be easily expressed by the unknowns $\textbf{u}$, we neglect the detailed analysis for those terms but give the lemma to conclude the analysis of these source terms.
\begin{lemma}\label{lemma:3.12}
Assume that $\|\textbf{u}\|_{L^{\infty}(\mathbb T^{3})}$ is sufficiently small, then we have the following equivalent relationship.
\begin{eqnarray*}
\hat{M}^{0\mu}&\sim& 2\tau \textbf{u}^{00}\del^{2}_{\tau}w^{2}\delta^{0}_{\mu}-\frac{4\dot{w}}{\tau}\textbf{u}^{00}\delta^{\mu}_{0}-\frac{\dot{w}}{\tau}(\textbf{u}^{0\mu}_{0}
+3\textbf{u}^{0\mu})
-\frac{4\dot{w}}{\tau}\textbf{u}^{0i}\delta_{0}^{(\mu}\delta^{0)}_{i}\nonumber\\
&&-\frac{2\dot{w}}{\tau}\textbf{u}^{0\mu}-4\del_{\tau}(w^{2}+\frac{\dot{w}}{2})(\textbf{u}^{0\mu}\delta^{\nu}_{0}+\textbf{u}^{0\nu}\delta^{\mu}_{0})\\
&&+(2\tau^{3\alpha+1}-2\tau^{6\alpha+4})A^{\frac{2}{\alpha+1}}I^{-\frac{\alpha}{\alpha+1}}H(\tau,\textbf{u})+H(\tau,\textbf{u})\\
&\sim&H(\tau,\textbf{u}),\\
\hat{M}^{ij}&\sim&-\frac{\dot{w}}{\tau}\del_{\tau}(g^{ij}-\eta^{ij})-\frac{\dot{w}}{\tau^{2}}(g^{ij}-\eta^{ij})\nonumber\\
&&\nonumber\\
&&+(2\tau^{3\alpha+1}-2\tau^{6\alpha+4})A^{\frac{2}{\alpha+1}}I^{-\frac{\alpha}{\alpha+1}}H(\tau,\textbf{u})+H(\tau,\textbf{u})\\
&\sim& H(\tau,\textbf{u}),\\
\acute{M}^{ij}&=&(det(\check{g}_{ab}))^{\frac{1}{3}}\textbf{L}_{lm}^{ij}\hat{M}^{lm}
-g^{\kappa\lambda}\del_{\kappa}((det(\check{g}_{ab}))^{\frac{1}{3}}\textbf{L}_{lm}^{ij})\del_{\lambda}(g^{lm}-\eta^{lm})\\
&\sim& H(\tau,\textbf{u}),\\
\hat{R}^{\textbf{q}}&\sim& \hat{M}^{00}-\frac{w^{2}}{3}g_{pq}\hat{M}^{pq}+\frac{4w\del_{\tau}w}{\tau}(2\tau \textbf{u}^{00}-\textbf{q})\nonumber\\
&&+g^{\kappa\lambda}(-\frac{w^{2}}{3}\del_{\kappa}g_{pq}\del_{\lambda}(g^{pq}-\eta^{pq})
-\frac{2w\del_{\tau}w\delta^{\kappa}_{0}}{3}g_{pq}\del_{\lambda}(g^{pq}-\eta^{pq})\\
&&-
\frac{2w\del_{\tau}w}{3}\delta_{0}^{\lambda}g_{pq}\del_{\kappa}(g^{pq}-\eta^{pq}))
+g^{\kappa\lambda}\left(\frac{2(\del_{\tau}w)^{2}+2w\del_{\tau}^{2}w}{w^{2}}\delta_{0}^{\lambda}\delta_{0}^{\kappa}(2\tau \textbf{u}^{00}-\textbf{q})\right)\\
&\sim&H(\tau,\textbf{u}).
\end{eqnarray*}
\end{lemma}
\begin{remark}\label{remark:3.13}
The regularity of above source terms with respect to $\tau$ can be derived directly from Corollary \ref{corollary:2.1}.
\end{remark}

Now what remains is to estimate the source term of the fluid equation $\hat{T}-\hat{\overline{T}}$.

\begin{lemma}\label{lemma:3.14}
Suppose that $\|\textbf{u}\|_{L^{\infty}(\mathbb T^{3})}$ is sufficiently small, we have
\begin{equation*}
\hat{T}-\hat{\overline{T}}\sim H(\tau,\textbf{u}).
\end{equation*}
\end{lemma}
\begin{proof}
As before, this can be derived by checking the following terms
\begin{eqnarray*}
\Gamma^{\kappa}(g_{\kappa0}f(\tau))&-&\overline{\Gamma}^{\kappa}(\eta_{\kappa0}f(\tau)),\vspace{2mm}\\
\tau^{3\alpha+2}(h^{\frac{1+\alpha}{2\alpha}}&-&\overline{h}^{\frac{1+\alpha}{2\alpha}}),\vspace{2mm}\\
B^{00}\del_{\tau}(g_{00}f(\tau))&-&\overline{B}^{00}\del_{\tau}(\eta_{00}f(\tau)),\vspace{2mm}\\
B^{0i}\del_{i}(g_{00}f(\tau))&-&\overline{B}^{0i}\del_{i}(\eta_{00}f(\tau)),\vspace{2mm}\\
B^{ij}\del_{i}(g_{0j}f(\tau))&-&\overline{B}^{ij}\del_{i}(\eta_{0j}f(\tau)),\vspace{2mm}\\
\frac{\del^{\mu}\Psi\del^{\nu}\Psi}{h}\Gamma^{\kappa}_{\mu\nu}(g_{\kappa0}f(\tau))&-&
\frac{\del^{\mu}\overline{\Psi}\del^{\nu}\overline{\Psi}}{\overline{h}}\overline{\Gamma}^{\kappa}_{\mu\nu}(\eta_{\kappa0}f(\tau)),\vspace{2mm}\\
\tau^{3(\alpha+1)}[h^{\frac{1-\alpha}{2\alpha}}\del^{\mu}\Psi\del^{\nu}\Psi\Gamma^{\kappa}_{\mu\nu}(g_{\kappa0}f(\tau))
&-&\overline{h}^{\frac{1-\alpha}{2\alpha}}\del^{\mu}\overline{\Psi}\del^{\nu}\overline{\Psi}\Gamma^{\kappa}_{\mu\nu}(\eta_{\kappa0}f(\tau)].
\end{eqnarray*}
Above seven terms are easy to analyze based on Lemmas \ref{lemma:3.6}-\ref{lemma:3.8}. We need to analyze the regularity of the following term with respect to $\tau$
$$
\del_{\tau}(f(\tau))=\del_{\tau}\left(\frac{(IK)^{\frac{\alpha}{\alpha+1}}w}{(I+K\tau^{3(\alpha+1)})^{\frac{\alpha}{\alpha+1}}}\right)=
-\frac{(IK)^{\frac{\alpha}{\alpha+1}}\dot{w}}{\tau w(I+K\tau^{3(\alpha+1)})^{\frac{\alpha}{\alpha+1}}}-
\frac{(IK)^{\frac{\alpha}{\alpha+1}}3\alpha Kw\tau^{3\alpha+2}}{(I+K\tau^{3(\alpha+1)})^{\frac{2\alpha+1}{\alpha+1}}}.
$$
According to Corollary \ref{corollary:2.1}, this term is regular when $\tau\rightarrow0$.
\end{proof}

We conclude sections \ref{section:3.2}-\ref{section:3.4} by an important proposition.

\begin{proposition}\label{prop:3.15}
Under the wave coordinates $Z^{\mu}=0$, the whole system \eqref{1.11} is equivalent to the following symmetric hyperbolic system
\begin{equation}\label{4.56}
\aligned
& B^{\mu}\del_{\mu}\textbf{u}=\frac{1}{\tau}\textbf{BP}\textbf{u}+H(\tau,\textbf{u})\;[1,0)\times\mathbb{T}^{3},
\\
&\textbf{u}=\textbf{u}_{1}\quad\quad\quad\qquad\quad \quad\quad\quad\text{in}\;{1}\times\mathbb T^{3}, 
\endaligned
\end{equation}
where $B^{\mu}$, $\textbf{B}$ and \textbf{P} are defined by \eqref{4.26}-\eqref{4.28}, \eqref{4.45} and satisfy the constraints of the general symmetric hyperbolic system discussed in Section \ref{section:3.5}. 
\end{proposition}

%----------------------------------------------------------------------------------------------------------------

\subsection{A class of symmetric hyperbolic systems}\label{section:3.5}

Consider the following symmetric hyperbolic system.
\begin{eqnarray}
B^{\mu}\del_{\mu}u&=&\frac{1}{t}\textbf{BP}u+H\quad\text{in}\;[T_{0},T_{1}]\times\mathbb{T}^{n},\\
u&=&u_{0}\quad\quad\quad\qquad\text{in}\;{T_{0}}\times\mathbb T^{n},
\end{eqnarray}
where

(i)\;$T_{0}<T_{1}\leq0$,

(ii)\;$\textbf{P}$ is a constant, symmetric projection operator, i.e., $\textbf{P}^{2}=\textbf{P}$, $\textbf{P}^{T}=\textbf{P}$,

(iii)\;$u=u(t,x)$ and $H(t,u)$ are $\mathbb R^{N}$-valued maps, $H\in C^{0}([T_{0},0],C^{\infty}(\mathbb R^{N}))$ and satisfies $H(t,0)=0$,

(iv)\; $B^{\mu}=B^{\mu}(t,u)$ and $\textbf{B}=\textbf{B}(t,u)$ are $\mathbb M_{N\times N}$-valued maps, and $B^{\mu},\,\textbf{B}\in C^{0}([T_{0},0],C^{\infty}(\mathbb R^{N}))$ and they satisfy
$$
(B^{\mu})^{T}=B^{\mu},\quad [\textbf{P}, \textbf{B}]=\textbf{PB}-\textbf{BP}=0,
$$

(v)\,there exists constants $\kappa,\,\gamma_{1},\,\gamma_{2}$ such that
$$
\frac{1}{\gamma_{1}}\mathbb I\leq B^{0}\leq \frac{1}{\kappa}\textbf{B}\leq\gamma_{2}\mathbb I
$$
for all $(t,u)\in[T_{0},0]\times\mathbb R^{N}$,

(vi)\;for all $(t,u)\in[T_{0},0]\times\mathbb R^{N}$, we have
 $$
 \textbf{P}^{\bot}B^{0}\textbf{P}=\textbf{P}B^{0}\textbf{P}^{\bot}=0,
$$
where
$
\textbf{P}^{\bot}=\mathbb I-\textbf{P}
$
is the orthogonal projection operator.

We will be able to conclude our argument with tthe help of the following result, whose proof relies on standard energy estimates; see  \cite{O}.

\begin{proposition}\label{pro:3.16}
Suppose that $k\geq \frac{n}{2}+1$, $u_{0}\in H^{k}(\mathbb T^{n})$ and assumptions (i)-(vi) are fulfilled. Then there exists a $T_{\ast}\in (T_{0},0)$, and a unique classical solution $u\in C^{1}([T_{0},T_{\ast}]\times\mathbb T^{n})$ that satisfies $u\in C^{0}([T_{0},T_{\ast}],H^{k})\cap C^{1}([T_{0},T_{\ast}],H^{k-1})$ and the energy estimate
$$
\|u(t)\|_{H^{k}}^{2}-\int_{T_{0}}^{t}\frac{1}{\tau}\|\textbf{P}u\|_{H^{k}}^{2}\leq
Ce^{C(t-T_{0})}(\|u(T_{0})\|_{H^{k}}^{2})
$$
for all $T_{0}\leq t<T_{\ast}$, where
$
C=C(\|u\|_{L^{\infty}([T_{0},T_{\ast}),H^{k})},\gamma_{1},\gamma_{2},\kappa),
$
and can be uniquely continued to a larger time interval $[T_{0},T^{\ast})$ for all $T^{\ast}\in(T_{\ast},0]$ provided $\|u\|_{L^{\infty}([T_{0},T_{\ast}),W^{1,\infty})}<\infty$.

Moreover, there exists a $\delta>0$ such that if $\|u_{0}\|_{H^{k}}\leq \delta$, then the solution exists on the time interval $[T_{0},0)$ and can be uniquely extended to $[T_{0},0]$ as an element of $C^{0}([T_{0},0],H^{k-1})$ satisfying
\begin{eqnarray*}
\|\textbf{P}u(\tau)\|_{H^{k-1}}\leq C\delta\left\{\begin{array}{ll}
-t &\text{if}\,\kappa>1,\\
t\ln(\frac{t}{T_{0}}) & \text{if}\, \kappa=1,\\
(-t)^{k}&\text{if}\, \kappa<1
\end{array}
\right.
\end{eqnarray*}
and
\begin{eqnarray*}
\|\textbf{P}^{\bot}u(\tau)-\textbf{P}^{\bot}u(0)\|_{H^{k-1}}\leq C\delta\left\{\begin{array}{ll}
-t &\text{if}\,\kappa\geq1\,\text{or}\,[B^{0},\textbf{P}]=0\\
-t+(-t)^{2\kappa} & \text{if}\, \kappa<1
\end{array}
\right.
\end{eqnarray*}
for $T_{0}\leq t\leq 0$.
\end{proposition}

%-------------------------------------------------------------------------------------------------------------

\subsection{Completion of the proof of the main result}\label{section:3.6}

Based on the analysis of subsections \ref{section:3.1}-\ref{section:3.5}, we have transformed the Einstein-Euler equations of GCG  into a symmetric hyperbolic system \eqref{4.56}. After a simple coordinate transformation $\tau\rightarrow -\tau$, the derived system is just the one considered in Section \ref{section:3.5}, thus we can use the main result of the Section \ref{section:3.5} to get Theorems \ref{theorem:1.2} and \ref{theorem:1.4} directly under the assumption that the matrix $-B^{0}$ of subsection \ref{section:3.3} is positive definite. Thus, we need to prove that $-B^{0}$ is positive definite based on the a priori estimates of Proposition \ref{pro:3.16}.

Obviously, from the smallness of the unknowns of Proposition \ref{pro:3.16}, we have the following equivalent relationships
\begin{eqnarray}\label{4.57}
h&=&-g_{\mu\nu}\del^{\mu}\Psi\del^{\nu}\Psi\sim -g_{00}f^{2}(\tau)\sim\overline{h},\\\label{4.58}
\del^{\tau}\Psi&=&g^{0\mu}\del_{\mu}\Psi\sim \del^{\tau}\overline{\Psi}\sim f(\tau),\\\label{4.59}
\del^{i}\Psi&=&g^{i\mu}\del_{\mu}\Psi\sim C\epsilon.
\end{eqnarray}
Thus, we have
\begin{eqnarray*}
-B^{00}&\sim&-(I-(\tau^{2(\lambda-1)}\overline{h})^{\frac{1+\alpha}{2\alpha}})(1+\frac{1-\alpha}{\alpha})\eta^{00}\nonumber\\
&&+\frac{1}{\alpha}(\tau^{2(\lambda-1)}\overline{h})^{\frac{1-\alpha}{2\alpha}}\tau^{2(\lambda-1)}(f(\tau))^{2}>0,\\
B^{jk}&\sim&(I-(\tau^{2(\lambda-1)}h)^{\frac{1+\alpha}{2\alpha}})\left(\delta^{jk}-\frac{1-\alpha}{\alpha}\frac{C^{2}\epsilon^{2}}{-g_{00}f^{2}(\tau)}\right)\nonumber\\
&&-\frac{1}{\alpha}(\tau^{2(\lambda-1)}\overline{h})^{\frac{1-\alpha}{2\alpha}}\tau^{2(\lambda-1)}C^{2}\epsilon^{2}
\sim(I-(\tau^{2(\lambda-1)}\overline{h})^{\frac{1+\alpha}{\alpha}})\delta^{jk}.
\end{eqnarray*}
Clearly, from \eqref{4.57}-\eqref{4.59} and Remarks \ref{remark:2.1}, \ref{remark:2.3}, $-B^{0}$ is a positive definite matrix. Hence, all the assumptions of Section \ref{section:3.5} are satisfied. Then we have by Proposition \ref{pro:3.16}
$$
\|\textbf{u}(\tau)\|_{H^{k}}\leq C\epsilon.
$$
By Lemma \ref{lemma:3.8}, we can equivalently get
\begin{equation}\label{4.60}
\|g^{\mu\nu}(\tau)-\eta^{\mu\nu}(\tau)\|_{H^{k+1}}+\|\del_{\kappa}g^{\mu\nu}(\tau)-\del_{\kappa}\eta^{\mu\nu}(\tau)\|_{H^{k}}
+\|\del_{\mu}\Psi(\tau)-\del_{\mu}\overline{\Psi}(\tau)\|_{H^{k}}\leq C\epsilon.
\end{equation}
Moreover, based on \eqref{4.56}, we have the following asymptotic behavior:
\begin{eqnarray}\label{4.61}
\|\textbf{P}(\textbf{u}^{0\mu}_{0}(\tau),\textbf{u}_{i}^{0\mu}(\tau),\textbf{u}^{0\mu}(\tau))^{T}\|_{H^{k-1}}&\leq& -C\epsilon \tau ln(\tau),\\\label{4.62}
\|\Pi(\textbf{u}_{0}^{ij}(\tau),\textbf{u}_{k}^{ij}(\tau),\textbf{u}^{ij}(\tau))^{T},\Pi(\textbf{u}_{0}(\tau),\textbf{u}_{k}(\tau),
\textbf{u}(\tau))^{T}\|_{H^{k-1}}&\leq&C\epsilon\tau.
\end{eqnarray}
And
\begin{eqnarray}\label{4.63}
\|(\mathbb I-\textbf{P})(\textbf{u}^{0\mu}_{0},\textbf{u}_{i}^{0\mu},\textbf{u}^{0\mu})^{T}(\tau)-(\mathbb I-\textbf{P})(\textbf{u}^{0\mu}_{0}(0),\textbf{u}_{i}^{0\mu}(0),\textbf{u}^{0\mu}(0))^{T}\|_{H^{k-1}}&\leq&C\epsilon \tau,\\\label{4.64}
\|(\mathbb I-\Pi)(\textbf{u}_{0}^{ij},\textbf{u}_{k}^{ij},\textbf{u}^{ij})^{T}(\tau)-(\mathbb I-\Pi)(\textbf{u}_{0}^{ij}(0),\textbf{u}_{k}^{ij}(0),\textbf{u}^{ij}(0))^{T}\|_{H^{k-1}}&\leq&C\epsilon \tau,\\\label{4.65}
\|(\mathbb I-\Pi)(\textbf{u}_{0},\textbf{u}_{k},\textbf{u})^{T}(\tau)-(\mathbb I-\Pi)(\textbf{u}_{0}(0),\textbf{u}_{k}(0),\textbf{u}(0))^{T}\|_{H^{k-1}}&\leq&C\epsilon \tau,\\\label{4.66}
\|\textbf{p}_{\mu}-\textbf{p}_{\mu}(0)\|_{H^{k-1}}&\leq&C\epsilon \tau.
\end{eqnarray}
Where in \eqref{4.61} and \eqref{4.62}, we have used $\kappa=1$ and $\kappa=2$ respectively.
From \eqref{4.60}-\eqref{4.66}, we have
\begin{eqnarray*}
\|\del_{\tau}(g^{0\mu}(\tau)-\eta^{0\mu}(\tau))-\frac{1}{\tau}(g^{0\mu}(\tau)-\eta^{0\mu}(\tau))\|_{H^{k-1}}+\|\del_{i}g^{0\mu}(\tau)\|_{H^{k-1}}&\leq&-C\epsilon \tau ln(\tau),\\
\|\del_{\tau}\textbf{q}(\tau)\|_{H^{k-1}}+\|\del_{\tau}\textbf{g}^{ij}(\tau)\|_{H^{k-1}}&\leq&C\epsilon \tau,\\
\|\del_{\tau}(g^{0\mu}(\tau)-\eta^{0\mu}(\tau))-\frac{2}{\tau}(g^{0\mu}(\tau)-\eta^{0\mu}(\tau))+\gamma^{\mu}\|_{H^{k-1}}&\leq&C\epsilon \tau,\\
\|\del_{i}\textbf{q}(\tau)-\del_{i}\textbf{q}(0)\|_{H^{k-1}}+\|\textbf{q}(\tau)-\textbf{q}(0)\|_{H^{k-1}}
\\+\|\del_{i}\textbf{g}^{ij}(\tau)
-\del_{i}\textbf{g}^{ij}(0)\|_{H^{k-1}}+\|\textbf{g}^{ij}(\tau)-\textbf{g}^{ij}(0)\|_{H^{k-1}}&\leq&C\epsilon \tau,\\
\|\del_{\mu}\Psi(\tau)-\del_{\mu}\Psi(0)\|_{H^{k-1}}&\leq&C\epsilon \tau,\\
\|\del_{\mu}\Psi(0)-\del_{\mu}\overline{\Psi}(0)\|_{H^{k-1}}&\leq&C\epsilon,\\
\|\del_{i}\Psi(0)\|_{H^{k-1}}&\leq&C\epsilon,
\end{eqnarray*}
where we have set
$$\gamma^{\mu}=\lim_{\tau\rightarrow0}-
[\del_{\tau}(g^{0\mu}(\tau)-\eta^{0\mu}(\tau))-\frac{2}{\tau}(g^{0\mu}(\tau)-\eta^{0\mu}(\tau))]
=\lim_{\tau\rightarrow 0}(\textbf{u}^{0\mu}(\tau)-\textbf{u}_{0}^{0\mu}(\tau)).
$$
On the other hand, we have
$$
\rho\sim\frac{A^{\frac{2}{\alpha+1}}I^{\frac{1}{\alpha+1}}}{\left(I-\tau^{3(\alpha+1)}(-\eta_{00}f^{2}(\tau))^{\frac{\alpha+1}{2\alpha}}\right)^{\frac{1}{\alpha+1}}}.
$$
Thus, when $\tau\rightarrow0$, we have
$\rho\rightarrow A^{\frac{2}{\alpha+1}}$
and
$
p_{\rho}=\frac{\alpha A^{2}}{\rho^{\alpha+1}}\rightarrow \alpha.
$
The proof of Theorem \ref{theorem:1.4} is completed.

%==================================================================

\section*{Acknowledgement}

This paper was written  in the Fall 2015 
at the Institut Henri Poincar\'e (IHP, Paris) during the Trimester Program ``Mathematical General Relativity'', which was organized by L. Andersson, S. Klainerman, and P.G. LeFloch.
The authors would like to thank Chao Liu for helpful discussions on the subject of this paper.
The first author (PGLF) gratefully acknowledges support from the Agence Nationale de la Recherche through grant ANR SIMI-1-003-01. The second author (CW) is grateful to the Fondation des Sciences Math\'{e}matiques de Paris (FSMP) for financial support during his stay at IHP.

%==================================================================

\end{document}